\PassOptionsToPackage{dvipsnames}{xcolor}
\documentclass[preprint]{elsarticle}
\usepackage{amssymb}
\usepackage{mathtools}
\usepackage{rotating}
\usepackage{makecell} 
\usepackage[utf8]{inputenc}
\usepackage{xspace}
\usepackage{subcaption}
\usepackage{amsmath}
\usepackage{verbatim}
\usepackage{listings}
\usepackage{icomma}
\usepackage{graphicx}
\usepackage{pifont}
\usepackage{lipsum}
\usepackage[ruled,vlined,linesnumbered,norelsize,Algorithm]{algorithm}
\usepackage[noend]{algpseudocode}
\usepackage{multirow}
\usepackage{relsize}
\usepackage[normalem]{ulem}
\usepackage{amsthm}
\usepackage[section]{placeins}
\usepackage{nameref, hyperref}

\usepackage{caption}
\usepackage{booktabs}
\usepackage{varwidth}

\makeatletter
\let\orgdescriptionlabel\descriptionlabel
\renewcommand*{\descriptionlabel}[1]{%
  \let\orglabel\label
  \let\label\@gobble
  \phantomsection
  \edef\@currentlabel{#1\unskip}%
  \let\label\orglabel
  \orgdescriptionlabel{#1}%
}
\makeatother

\usepackage{svg}

\theoremstyle{definition}
\newtheorem{lemma}{Lemma}

\newtheorem{theorem}{Theorem}
\newtheorem{definition}{Definition}

\newtheorem{corollary}{Corollary}
\newtheorem{innercustomgeneric}{\customgenericname}
\providecommand{\customgenericname}{}
\newcommand{\newcustomtheorem}[2]{%
  \newenvironment{#1}[1]
  {%
   \renewcommand\customgenericname{#2}%
   \renewcommand\theinnercustomgeneric{##1}%
   \innercustomgeneric
  }
  {\endinnercustomgeneric}
}
\newcustomtheorem{customthm}{Theorem}
\newcustomtheorem{customlemma}{Lemma}
\newcustomtheorem{customcorollary}{Corollary}

\renewcommand{\textsf}{\texttt}

\usepackage{tcolorbox}

\newcommand{\algorithmname}{Query and Parallelism Optimized Space-Saving\xspace}
\newcommand{\subalgorithmname}{Query Optimized Space-Saving\xspace}
\newcommand{\shortalgorithmname}{QPOPSS\xspace}
\newcommand{\shortsubalgorithmname}{QOSS\xspace}

\newcommand{\name}{\emph{\algorithmname}\xspace}

\newcommand{\newtext}[1]{#1}

\usepackage{natbib}
\usepackage{graphicx}
\usepackage{enumitem}

\begin{document}
\begin{frontmatter}

\title{\shortalgorithmname: \name for 
Finding Frequent Stream Elements
}
\author[1,2]{Victor Jarlow\texorpdfstring{\corref{cor1}}}
\ead{victor.jarlow@\{astazero.com,ri.se,outlook.com\}}
\author[1,3]{Charalampos Stylianopoulos}
\ead{charalampos.stylianopoulos@gmail.com}
\author[1]{Marina Papatriantafilou}
\ead{ptrianta@chalmers.se}
\cortext[cor1]{Corresponding author, whose majority of work was carried out while affiliated with Chalmers University of Technology}
\affiliation[1]{organization={Dept. Computer Science and Eng., Chalmers Un. of Technology and Un. of Gothenburg, Gothenburg, Sweden}}
\affiliation[2]{organization={R\&D Dept., AstaZero, Gothenburg, Sweden}}
\affiliation[3]{organization={emnify, Berlin, Germany}}

\begin{abstract}
The frequent elements problem, a key component in demanding stream-data analytics, involves selecting elements whose occurrence exceeds a user-specified threshold.
Fast, memory-efficient $\epsilon$-approximate {synopsis algorithms} select all frequent elements but may overestimate them depending on $\epsilon$ (user-defined parameter). Evolving applications demand performance only achievable by parallelization. However, algorithmic guarantees concerning concurrent updates and queries have been overlooked.
We propose 
\algorithmname (\shortalgorithmname), 
providing concurrency guarantees. The design includes  
an implementation of the \emph{Space-Saving} algorithm supporting fast queries, implying minimal overlap with concurrent updates. 
\shortalgorithmname integrates this with the distribution of
work and fine-grained synchronization among threads,  swiftly balancing high throughput, high accuracy, and low memory consumption. 
Our analysis, under various concurrency and data distribution conditions, shows space and approximation bounds. Our empirical evaluation relative to representative state-of-the-art methods reveals that \shortalgorithmname's multi-threaded throughput scales linearly while maintaining the highest accuracy, with orders of magnitude smaller memory footprint.

\end{abstract}

\end{frontmatter}

\section{Introduction}\label{sec:introduction}

Efficient data synopsis algorithms are a core part of many important applications, including the online processing of events, click-streams, web log analysis, natural language processing, heavy flow detection in computer networks, dimensionality reduction in Machine Learning (ML) applications and more~\cite{charikar_finding_2004, cormode-roadmap2021, hernandez2020eraia, may_streaming_2017, metwally_efficient_2005}.
E.g., in a stream of web page user clicks from a vast number of users, synopsis queries may estimate the number of
unique users clicking on a particular link, the most active users, quantiles describing the time spent on a web page, and more.
Synopses give (often very close to accurate) answers to queries while having small memory footprint and processing times. 

Given the increasingly high data rates due to increased digitalization globally, data cannot be processed continuously in a scalable and timely manner without parallelism. In the literature, there is an established volume of knowledge on data synopses~\cite{cormode_finding_2008, charikar_finding_2004,estan_new_2003, HeavyHittersP4,manerikar_frequent_2009,metwally_efficient_2005}. However, studies regarding consistency implications of algorithmic implementations enabling concurrency of queries and updates are relatively recent, with few results in the literature, also describing the need and significance of addressing the problem in practice as well as analytically~\cite{cormode2018data, stylianopoulos_delegation_2020, RinbergKeidarDisc2020, rinberg2023intermediate}.
This gap in the literature means that the way such algorithms perform (in terms of processing timeliness and accuracy) in a parallel execution environment is widely unclear.

In this work, we target the \emph{frequent elements} problem, in which data synopses are queried to respond with the elements whose occurrence in the stream exceeds a user-specified threshold. The problem has multiple applications, including continuous monitoring, data and ML pipelines, as discussed in related literature~\cite{cafaro_parallel_2016, das_cots:_2009, heavy-keeper, TopKapi, Augmented_Sketch,top-k-sgd,zhang_efficient_2014}. 
Consider one of the more intuitive applications in continuous monitoring, e.g. of network traffic, where an IP network flow identifies a connection and is represented by a 5-tuple of source and destination IP addresses, source and destination port numbers, and protocol. A small number of distinct flows, dubbed \emph{elephant flows}, tend to make up a large share of the bandwidth consumption, and tracking them is useful for accounting and statistics \cite{goos_frequency_2002}, detecting anomalous patterns such as DDoS attacks~\cite{harrison_network-wide_2018}, as well as for dynamically scheduling network traffic in Software Defined Networks~\cite{al-fares_hedera_2010}.
An exact answer requires tracking each unique flow, 
and consumes memory space on the order of exabytes due to the massive number of bit-combinations possible in a flow, which is an impractical approach. If a small and controllable incurred error is acceptable, the memory consumed can be drastically reduced by using a synopsis data structure.
Moreover, high-velocity streams limit the per-packet processing time (ppt). In the case of backbone routers, which process vast amounts of data, e.g., the optical carrier bandwidth specifications OC-192 and OC-782 can require a ppt of less than 100 ns and 25 ns respectively~\cite{kumar_data_2004}. Achieving such low ppt, in the presence of a high incoming packet rate,  is in some cases only possible by utilizing multiple processor cores in parallel.
 
Much of the previous work on \emph{parallelizing frequent elements algorithms} \cite{parallel_space_saving,cafaro_parallel_2016, TopKapi} addresses performance on only updates without evaluating the effect of concurrent queries.
Understanding performance under concurrent updates and queries is essential, as the answer to a streaming query may naturally be needed on the fly and regularly, preferably without halting the processing of the high-rate stream, balancing inherent trade-offs between timeliness and accuracy. 

To satisfy this need, we make the following contributions:

\begin{itemize}[topsep=0pt, leftmargin=*]
\setlength\itemsep{0em}
    \item We build on the Space Saving algorithm \cite{metwally_efficient_2005} to propose Query and Parallelism Optimized Space-Saving (QPOPSS),
a memory-conservative, scalable method for estimating frequent
elements that support concurrent updates and queries while maintaining
high approximation-accuracy results. 
 \item A key to achieving the above is our  \emph{\subalgorithmname} (\shortsubalgorithmname) algorithmic implementation, which enhances known algorithm designs by supporting lower-latency bulk operations (queries for the Space-Saving algorithm), thus enabling reduced update overlaps and therefore improved accuracy preservation under concurrency.
    \item We provide a detailed analysis of the \shortalgorithmname properties under various concurrency and data distribution conditions, showing space and approximation bounds.
    Moreover, we analyze the problem regarding efficiency and concurrency-aware accuracy-associated trade-offs, which, to our knowledge, have not been done before.
    \item We present a detailed evaluation of our open-source implementation \cite{Delegation_Space-Saving_2022} using high-rate data from real-world and synthetic sources to investigate the associated trade-offs. The multi-threaded \shortalgorithmname is studied together with \shortsubalgorithmname as well as the representative alternatives Topkapi \cite{TopKapi} and PRIF \cite{zhang_efficient_2014}
    regarding parallelism speedup,  query accuracy and latency, memory consumption, as well as overall processing throughput. Our design uses as little as almost the same amount of memory as the single-threaded version, sometimes consuming $10^{-4}$ less memory bytes than other approaches. Still, \shortalgorithmname shows very high accuracy and linear speedup.
\end{itemize}

The rest of the paper is structured as follows: Section~\ref{sec:preliminaries} covers the background, a description of the system model, and the basic metrics of interest. Section~\ref{sec:problem_analysis} analyzes the problem at hand, motivating a novel balanced approach, and outlines the associated challenges relative to state of the art. For \shortalgorithmname, we give an overview and its algorithmic implementation in Section~\ref{sec:deleg_space_saving}, while 
Sections~\ref{sec:algo_impl_analysis} and~\ref{sec:evaluation} cover its analysis and empirical evaluation. We discuss other related work and present our conclusions in Sections~\ref{sec:relatedwork} and~\ref{sec:conclusions}, respectively.\looseness=-1

\section{Preliminaries}\label{sec:preliminaries}

In an unbounded stream of elements $\mathcal{S}\coloneqq \tau_1,\tau_2,...,\tau_N, ...$ 
for any $N$, we say that an element $e\in U$ has a frequency count $f_N(e) = |\{ j | \tau_j = e\}|$ after $N$ elements have been processed, where $U$ is the universe of possible elements. 
From this point on, we assume that $\mathcal{S}$ contains only elements of positive unit weight, also known as the cash-register model~\cite{GarofalakisDataStreamManagement}.

The $\phi$-\emph{frequent elements problem} is concerned with selecting the elements of a stream with a frequency count above $\phi N$, where $\phi\in [0,1]$. 
To find the $\phi$-frequent elements of an arbitrary stream using a deterministic algorithm, at least $O(|U|)$ space has to be used \cite{cormode_finding_2008}.
For applications where an approximation is acceptable, we can relax the exact problem to the $\epsilon$\emph{-approximate} $\phi$-\emph{frequent elements problem}\footnote{From this point on, we may refer to the $\epsilon$\emph{-approximate} $\phi$-\emph{frequent elements problem} as the frequent elements problem.}, defined in~\cite{manku_approximate_2002} as follows.

\begin{definition}
 \label{def:epsapprox}
 Given a stream $\mathcal{S}$ of $N$ elements, 
 the \emph{$\epsilon$-approximate $\phi$-frequent elements} problem is to \emph{report} a set $F$ containing all elements $e\in U$ with $f_N(e) > N\phi$ and no element $e \in U$ with $f_N(e) < N(\phi - \epsilon)$, where $0<\epsilon\leq \phi<1$.
\end{definition}
Accordingly, $F$ must contain all elements with frequency of occurrence higher than $N\phi$ but may also include elements that occur at least $N(\phi - \epsilon)$ times due to approximation-induced errors.

Also specified in ~\cite{manku_approximate_2002}, a closely related definition concerns the estimated frequency count of an individual element.

\begin{definition}
\label{def:frequencyestimation}
Given a stream $\mathcal{S}$ of $N$ elements, an \emph{$\epsilon$-approximate frequency estimation} denoted $\hat{f_N}(e)$, of the true frequency count of $e\in U$ is bounded: $f_N(e) \leq \hat{f_N}(e) \leq f_N(e) + N\epsilon$. 
\end{definition}
In other words, the estimated count of an element is bounded from above by the sum of the elements' actual frequency count and a fraction of the stream length, as decided by the $\epsilon$-factor.

Algorithms that report the  $\epsilon$-approximate $\phi$-frequent elements described in definitions \ref{def:epsapprox} and \ref{def:frequencyestimation} generally provide at least two operations:
\begin{itemize}[leftmargin=*]
\setlength\itemsep{0em}
    \item \textbf{\emph{Update(e,w)}}: element $e$ is processed, registering its occurrence in the stream. If \emph{weighted} updates are supported, then $w$ is the number of simultaneous arrivals of $e$ to update the data structure with.
    \item \textbf{\emph{Query($N,\phi $)}}: Returns $F$ from definition \ref{def:epsapprox} with estimated frequency count of each individual element $e \in F$ adhering to the bounds in definition \ref{def:frequencyestimation}.
\end{itemize}

These algorithms can be distinguished into two classes:

\noindent\textbf{Counter-based Algorithms:} which are generally deterministic and keep a fixed-size set that contains tuples of an element and its estimated occurrence in the stream. As an individual element is observed, its associated estimated occurrence is incremented. The sets' fixed size demands a tactic for managing the occurrence of an element not in the set while the set is full. The tactic often leads to an incurred error but can be chosen such that the error is minimized depending on the area of use. Keeping more counters reduces the error and vice versa, highlighting the \emph{memory space/error trade-off} present in all approximate synopsis data structures. 
Prominent such algorithms are the \emph{Frequent Algorithm} (also known as the \emph{Misra-Gries Algorithm}) \cite{misra_finding_1982,goos_frequency_2002,karp_simple_2003}, \emph{Lossy Counting} \cite{manku_approximate_2002}, and \emph{Space-Saving} \cite{metwally_efficient_2005}.

\noindent\textbf{Sketch-based Algorithms:} 
Use randomized hashing to compress the stream state into a 1- or 2-dimensional array of counters. 
Updating the sketch involves computing a hash value for the incoming element. A counter corresponding to the hash value is then, e.g., incremented or decremented, concluding the update operation. Since the same hash value can be computed for multiple different elements (i.e., hash collisions), the compressed stream state is encoded in the counters. Element queries are carried out by computing some statistic, e.g., minimum or median, on the counters mentioned above. Heap data structures can be used to supplement sketches and track frequent elements. Similar to counter-based algorithms, sketch-based ones also imply a memory space/error trade-off: increasing the memory consumed by the sketch can reduce the error and also increase the probability of remaining within the error bounds 
and vice versa. Influential sketch-based algorithms that can form the components of a solution to the frequent elements problem include the \emph{Count-Min Sketch} \cite{cormode_improved_2005} and \emph{Count Sketch} \cite{charikar_finding_2004}.\looseness=-1

\noindent
\textbf{Comparison:}
\label{sec:comparison_spacesaving}
Counter-based algorithms demonstrate superior accuracy per memory byte when processing a continuous stream of positive updates compared to sketch-based algorithms~\cite {hsu2019learning}. Moreover, Space-Saving has been shown to guarantee better accuracy than both Frequent and Lossy Counting while having better or equivalent throughput compared to both of the latter~\cite{cormode_finding_2008}. 
\looseness=-1

\noindent\textbf{Concurrency model:}
For our algorithm design, we consider a system with multiple cores and a set of sequential threads $t_{1},...,t_{T}$, that does not arbitrarily fail or halt, capable of communicating using asynchronous shared memory, supported by a coherent caching model.
Threads in the system can perform one of two operations at a time: updates by consuming elements from the input stream or responding to a frequent elements query.

\noindent\textbf{Performance metrics:} Metrics for synopsis algorithms include both \emph{time/space efficiency} and \emph{error} ~\cite{cormode_finding_2008,manerikar_frequent_2009,metwally_efficient_2005}. Key time/space metrics are \emph{latency} (duration of operations), \emph{throughput} (number of updates or queries carried out per unit of time), \emph{scalability} (the ability to utilize efficiently multiple threads), and \emph{consumed memory}.
Regarding error metrics, the aforementioned $\epsilon$ factor can be tuned to impact the \emph{precision} (fraction of 
 relevant elements reported out of all reported elements), \emph{recall} (fraction of relevant elements reported out of all relevant elements), and \emph{average relative error} (the average of all per-element absolute estimation errors divided by each actual in-stream occurrence). Further, we need to consider the effects of concurrency on estimation error, which we discuss in the following section, where we analyze the problem relating to key trade-offs.
 
\section{Problem analysis}
\label{sec:problem_analysis}
Several works present designs that utilize parallelism for the frequent elements problem \cite{zhang_efficient_2014, TopKapi,cafaro_parallel_2016,das_cots:_2009, Augmented_Sketch}. The previous designs can be distinguished into those based on \emph{Global Data Structure} and \emph{Thread-local Data Structures}. 
We discuss their associated trade-offs and describe the problem of query accuracy as it relates to concurrency. Lastly, we motivate the need for a new approach.

\subsection{Global Data Structure}
This category's design features a singular synopsis data structure that all threads access through query and update operations, making efficient synchronization methods an integral part of the algorithm design.
The memory footprint is similar to sequential frequent elements algorithms, with any excess memory being attributed to inter-thread synchronization mechanisms.
The main benefit is that a query can be answered directly by accessing the singular synopsis tracking the $\epsilon$-approximate $\phi$-frequent elements.

The {Cooperative Thread Scheduling Framework (CoTS)}~\cite{das_cots:_2009} and its multi-stream extension~\cite{parallel_space_saving} belong to this category. 
The design features a method for enabling thread cooperation rather than contention, achieved by threads getting help in accomplishing their update operations by the one thread currently accessing the data structure at any instant. 
The works lack discussion on the effect overlapping updates and queries have on query accuracy and rely on the analysis of the sequential Space-Saving algorithm.

\subsection{Thread-Local Data Structures}
The thread-local data structures category contains designs utilizing several synopsis data structures, one for each thread. Since each thread only updates a local synopsis data structure, the update rate scales well with the number of threads and can be carried out without synchronization.
Due to duplication of data structures, designs in this category fail to maintain the same accuracy guarantees as a sequential algorithm without consuming memory space on par with a sequential solution multiplied by the number of threads.
Moreover, to perform a query, a thread has to merge multiple synopses, which can become a predominant factor of latency when the number of threads is high.

The Topkapi sketch~\cite{TopKapi} and the parallel algorithm in~\cite{cafaro_parallel_2016} utilizing Space-Saving follow this approach; worth noting is that neither of them support overlapping queries and updates. The PRIF algorithm in \cite{zhang_efficient_2014} allows overlapping queries and updates. Its memory-intensive design features several thread-local data structures that periodically update a single large data structure containing the final synopsis. While the memory/accuracy trade-off is rigorously examined, query throughput and concurrency guarantees are not discussed.

\subsection{Accuracy and Consistency}\label{sec:concurrency-awareX}
Since synopsis queries target \emph{approximate} output,
it can be observed that strict consistency requirements and associated synchronization can induce an overly excessive, partly unnecessary overhead in the concurrent setting.

Since there is no, to our knowledge, concurrency-aware definition of the set of $\epsilon$-approximate $\phi$-frequent elements, we find it natural to explore how such a definition could be formulated and its implications.
Works on concurrency-aware semantics discuss notions such as {regularity}, {intermediate value linearizability}, $k$-out-of-order relaxation and more~\cite{nikolakopoulos2015consistency, RinbergKeidarDisc2020, rinberg2023intermediate, k-relaxations}.

These specifications model queries/operations on objects that return values complying with ``observing" associated subsets of update operations relative to the sets implieed by linearizability or sequential consistency.
Such concurrency models inspire further accuracy and consistency analysis of queries involving overlapping updates.

\subsection{Need for a Balancing Approach}
\label{sec:challenges}
The categories mentioned earlier represent two contrasting perspectives. Firstly, in the global-data structure approach, the frequent elements are tracked in a single data structure that can be scanned quickly, enabling high query throughput. However, the single data structure can become detrimental when processing high-velocity data streams due to the synchronization overhead involved in supporting concurrent multi-thread access. Secondly, the thread-local approach allows for high update throughput and scalability since multiple threads can process stream elements in parallel. However, the frequent elements are scattered across several distinct data structures that occupy precious bytes of memory and require latency-inducing assembly upon querying.

Clearly, there is a need for a concurrent approach to the frequent elements problem that balances and preferably combines, to a large extent, the favorable properties of both categories favorably: a high update/query throughput, low memory space, high accuracy, and low query latency solution to the highest possible extent. Moreover, an essential missing part in the literature is a thorough analysis (both theoretical and empirical) of the effects of overlapping updates and queries. This analysis can serve as a tool for further exploration of the involved trade-offs.
To this end, we identify a set of challenges in designing a balanced approach for the frequent elements problem, namely to enable:
\begin{description}
\setlength\itemsep{0em}
    \item[{[}C1{]}\label{challenge:throughput_latency}] 
    high query and update throughput without impacting query latency;    
    \item[{[}C2{]}\label{challenge:memory_accuracy}] 
    parallelism without an inherent impact on memory and accuracy;
    \item[{[}C3{]}\label{challenge:accuracy_bound}]reasoning about accuracy guarantees in the presence of overlapping updates and queries.
\end{description}

We discuss how we address these challenges in the following section, where we present our method. The properties of our proposed method regarding challenge \ref{challenge:accuracy_bound} are further discussed in Section \ref{sec:algo_impl_analysis}.
\section{\algorithmname}\label{sec:deleg_space_saving}
In this section, we present our proposed method, an accuracy-preserving multi-threaded design for finding the frequent elements of a stream, supporting concurrent updates and queries.
Due to the properties of Space-Saving as mentioned in Section~\ref{sec:preliminaries}, we have chosen it as a component in our design, hereafter referred to as~\algorithmname (\shortalgorithmname{}).

We begin with an overview of the \shortalgorithmname~design and its components along with some auxilliary concepts, followed by our sequential Space-Saving algorithmic implementation, with latency optimizations as motivated in Section~\ref{sec:challenges}, and finally, we describe both the update and query procedure of \shortalgorithmname in detail.

\subsection{Design Overview}
\label{sec:designOverivew}

In addressing the challenges listed in Section \ref{sec:challenges}, several design choices were made, which we outline here and connect to the method we are about to present.

To address challenge \ref{challenge:throughput_latency}, we present our proposal to reduce query latency and improve query throughput drastically, namely the \emph{\subalgorithmname} (\shortsubalgorithmname) algorithmic implementation, an integral part of our method that diminishes overlaps of queries with concurrent updates as a consequence of its design.
By making use of the min-max heap data structure~\cite{atkinson_min-max_1986}, we facilitate finding both the element with the least count, which is essential to the Space-Saving update procedure, and the elements with the largest counts, promoting high query throughput.
This is achieved as the min-max heap groups elements with similar counts together, allowing them to be swiftly selected during a query.

Challenge \ref{challenge:memory_accuracy}, which relates to maintaining accuracy and space guarantees in a concurrent setting, is addressed by designing our approach to adopt domain splitting \cite{stylianopoulos_delegation_2020}. The purpose of this scheme is to delegate responsibility for a subset of the domain of possible elements to each thread. As we show in Section~\ref{sec:algo_impl_analysis}, Space-Saving's accuracy and memory efficiency benefit greatly from this.
This adoption implies a new challenge regarding running a global query over parts of the data structure maintained by different threads; we explain how we address this through efficient synchronization in section \ref{sec:fequeries}.

To address challenge \ref{challenge:accuracy_bound} (in conjunction with \ref{challenge:throughput_latency}),
focusing on ensuring high concurrent update and query throughput with known accuracy guarantees, \shortalgorithmname supports both a) concurrent updating and b) concurrent querying. 

In particular, when it comes to a), to maintain high update throughput, we wish to minimize the need for costly synchronization between threads. To this end, we base our update method on a thread-cooperation technique~\cite{stylianopoulos_delegation_2020} that promotes thread-local updates to a large extent.
The technique involves associating a series of lightweight thread-local \emph{filter} data structures with each thread that acts as buffers for elements owned by other threads. This approach makes inter-thread synchronization necessary only when filters are full; at this point, the filter is transferred to the owner thread and fed as updates to a reserved \shortsubalgorithmname algorithm. 
We extend this technique to support \emph{global queries}, which we describe with proper context in more detail in the upcoming subsection~\ref{sec:updates}.

As for b), the objective is to collect the frequent elements from each thread's reserved \shortsubalgorithmname algorithm while causing as little contention as possible. Therefore, the query procedure is designed to be lightweight. This leads to minimal contention on the \shortsubalgorithmname algorithm data structures in memory in conjunction with buffered updates. Furthermore, the risk of overlapping thread access is reduced, promoting high throughput and low latency. We analyze the consistency-related implications of our query method in greater detail in Section \ref{sec:algo_impl_analysis}.
The method for updating and querying \shortalgorithmname is based on the auxiliary concepts presented in the following subsection.

\subsection{Auxiliary Concepts}

Besides query optimization that targets parallelism-aware accuracy, \shortalgorithmname builds on two concepts from \cite{stylianopoulos_delegation_2020}, which we reiterate here, for self-containment.

\noindent
\textbf{Domain Splitting}
logically divides $U$, the domain of possible elements, into equally sized subsets.  Ownership is then distributed to each of the $T$ dispatched threads using the function $owner: U \to \{1..T\}$, which maps an element from the input domain to a thread-id. 
partitions $U$, the domain of possible elements, into subsets and distributes ownership to each of the $T$ dispatched threads using the function $owner: U \to \{1..T\}$.
The subdomain of elements owned by thread $i$ is therefore $U_i = \{e \in U\ |\ owner(e) = i\}$. 

\noindent
\textbf{Delegation Filters} facilitate efficient inter-thread communication by buffering elements owned by other threads. Full filters are handed to the thread owning the contained elements. For each of the $T$ dispatched threads, a series of $T$ Delegation Filters are reserved for each thread, arranged in a $T \times T$ matrix. Delegation filters are small and of fixed size to maintain a low memory footprint. 

At this point, we depart from the approach taken in \cite{stylianopoulos_delegation_2020} and describe the foundations of our method in the following subsections, starting with our query-optimized Space-Saving algorithmic implementation in the following subsection.

\subsection{\subalgorithmname}
\label{sec:QOSS}
Emphasizing improved query processing timeliness, we propose our  \subalgorithmname (\shortsubalgorithmname) algorithmic implementation of Space-Saving. \shortsubalgorithmname retains the accuracy guarantees and memory requirements of  Space-Saving, while using optimized underlying data structures and query procedures. 

For self-containment, first 
we summarize  the {basics of Space-Saving} \cite{metwally_efficient_2005} for the $\epsilon$-approximate frequent elements problem: a set of $m$ tuples of the form $(e,\hat{f(e)})$ are kept. If an element $e$ that is in the set arrives, the associated estimated count $\hat{f(e)}$ is incremented. If $e$ is not in the set, and the set contains fewer than $m$ tuples, then  $(e,1)$ is added to the set; if the set contains $m$ tuples, the one with the least counter, $(e_{min},\hat{f(e)}_{\{min\}})$ is identified, and the tuple $(e_{new},\hat{f(e)}_{min}+1)$ takes its place, effectively replacing $e_{min}$ and incrementing the estimated count by 1. 
When queried, a set of tuples containing an estimated frequency count greater than $N\phi$ are output. 
Setting the number of counters to $m=\frac{1}{\epsilon}$ 
ensures that the set of $\epsilon$-approximate $\phi-$frequent elements (cf. Definitions \ref{def:epsapprox} and \ref{def:frequencyestimation}) is returned \cite{metwally_efficient_2005}.

Efficient algorithmic implementations of Space-Saving exist through the \emph{Space-Saving Linked List} (SSL) and \emph{Space-Saving Heap} (SSH) data structures, extensively evaluated in \cite{cormode_finding_2008}. 
SSL, similar to the \emph{Stream-Summary} in \cite{metwally_efficient_2005}, keeps a linked list sorted on the estimated frequency of elements. SSH keeps elements in a \emph{min-heap} of size $m$, which makes finding the element with the least count possible in $O(1)$ time. The hash map data structure is also time-efficient.
The evaluation argues that SSH is somewhat slower than SSL but requires significantly less memory space. Moreover, SSH supports \emph{weighted updates}, which SSL does not. Such updates allow several/weighted counts of the same element to take place at the same computational cost as a single element. This makes SSH the option of choice in the weighted setting~\cite{anderson_high-performance_2017}, also needed here.

However, notice a \emph{shortcoming}: to answer a query, an array of size $m$ has to be traversed in SSH. At each traversal, the element count is compared to a threshold value to decide if the element belongs to the output set, which requires time linear in $m$. 

To alleviate this shortcoming and {address \ref{challenge:throughput_latency} of Section~\ref{sec:challenges}},  the \emph{\subalgorithmname} (\shortsubalgorithmname) keeps a \emph{min-max heap}~\cite{atkinson_min-max_1986} data structure of size $m$ with  the counter count satisfying that: 
\begin{itemize}[topsep=0pt, leftmargin=*]
\setlength\itemsep{0em}
    \item at an even level (min-level) it is less than all of its descendants;
    \item  at an odd level (max-level) it is greater than all of its descendants.
\end{itemize}
The min-max heap allows: a) finding the least element in $O(1)$, which is essential to perform update operations quickly, and b) performing a query in $O(|F|)$ time, where $|F|$ is the number of frequent elements from definition \ref{def:epsapprox}. The number of elements in $F$ is commonly significantly less than $m$, especially when the input stream is skewed. This modification introduces a slight per-element processing overhead, overshadowed by the overall throughput gain when queries are repeatedly carried out while high-rate streams are processed.

\begin{figure}[tb!]
\begin{center}
    \includegraphics[width=0.70\columnwidth]{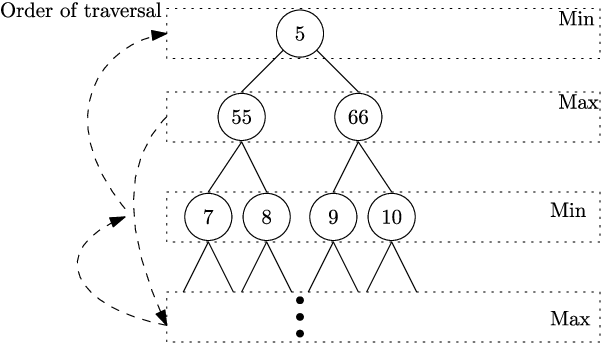}
\end{center}
    \caption{A binary min-max tree with alternating levels. The dashed arrows depict the traversal order during a \shortsubalgorithmname query.
    \label{fig:QOSS_traversal_order}
    }
\end{figure}

\begin{algorithm}[t!]
\small
\caption{Algorithmic implementation of the Query Optimized Space-Saving (QOSS) algorithm using a binary min-max heap 
\label{algo:qoss}
}
\begin{algorithmic}[1]
\Function{$InitializeQOSS$}{$\epsilon \in [0,1]$}
\State $m \gets \lceil\frac{1}{\epsilon}\rceil$
\State $m \gets 4 \lfloor \frac{m}{4} \rfloor + 3$ \Comment{All nodes have 3 or 0 grandchildren}
\State let $H$ be a min-max heap of size $m$ of counters initialized to ($\varnothing,0)$
\State let $M$ be a hash map of size $O(m)$ of pointers to counters
\EndFunction
\Function{$UpdateQOSS$}{$e \in U$, $w \in \mathbb{N}$}
\If {$(e,\hat{f(e)}) \in \textit{M}$}
    \State $i \gets $ \textit{M.Find(e)}
    \State $H[i] \gets (e,\hat{f(e)}+w)$ 
\Else 
    \State $(\_,\hat{f(e)}_{min}) \gets H[1]$
    \State $H[1] \gets (e,\hat{f(e)}_{min}+w)$
\EndIf
\State Ensure min-max heap property of $H$ is maintained
\EndFunction
\Function{$QueryQOSS$}{$\phi\in [0,1]$, $N\in\mathbb{N}$}
\State initialize empty stack
\State V $\gets \varnothing$
\State push 2 and 3 to stack
\While{stack is not empty}
    \State $i$ $\gets$ stack.pop()
    \State $(e,\hat{f(e)}) \gets H[i]$
    \If{$\hat{f(e)} >=\phi N$}
        \State Output $(e,\hat{f(e)})$
        \State traverse\_next\_level()
    \EndIf
\EndWhile

\EndFunction

\Function{$traverse\_next\_level$}{}
\If{$\lfloor \log_2(i)\rfloor \equiv 1 \mod 2$} \Comment{i is on a max-level}
    \If{$4i+3 <= m$} \Comment{i has grandchildren}
        \State push $4i + j, j\in \{0,1,2,3\}$ to stack
    \Else
        \If{$2i+1 <= m$} \Comment{i only has children}
            \State push $2i + j, j\in \{0,1\} $ to stack
        \Else \Comment{$i$ has no children or grandchildren}
            \State push $\lfloor \frac{i}{2} \rfloor$ to stack if $\lfloor \frac{i}{2} \rfloor \notin V$
            \State $V \gets V \cup \{\lfloor \frac{i}{2} \rfloor\} $ 
        \EndIf    
    \EndIf
\Else \Comment{i is on a min-level}
    \If{$\lfloor \frac{i}{4}\rfloor > 0$} \Comment{i has a grandparent}
        \State push $\lfloor \frac{i}{4}\rfloor$ to stack if $\lfloor \frac{i}{4}\rfloor \notin V$
        \State $V \gets V \cup \{\lfloor \frac{i}{4}\rfloor\}$
    \EndIf
\EndIf

\EndFunction

\end{algorithmic}
\end{algorithm}

Algorithm~\ref{algo:qoss} shows the \shortsubalgorithmname pseudocode and a description of it follows here.

\noindent\textbf{Initialization:} As can be seen in Algorithm~\ref{algo:qoss} lines 1 through 5, the number of counters, $m$, is calculated using the desired $\epsilon$-factor. In line 3, steps ensure that each node has either 3 or 0 grandchildren. A counter is a tuple of an element identity $e$ and an estimated count $\hat{f(e)}$. The $m$ counters are arranged in a tree structure organized by levels in the min-max heap $H$ (see Algorithm~\ref{algo:qoss} line 4). Counters can be accessed through a hash map $M$, initialized at line~5.

\noindent\textbf{Updates:} 
Lines 6 through 12 detail the update procedure. If $e$ is found through hash-table lookup, $\hat{f(e)}$ is incremented by $w$, the weight of the update. 
Otherwise, $e$ takes the place of the counter with the least count, $\hat{f(e)}_{min}$, which is incremented by $w$. 
If the min-max heap property of $H$ is broken, it will be restored by performing at most $O(log(m))$ swaps involving $(e,\hat{f(e)})$.

\noindent\textbf{Queries:}
The query procedure is designed to avoid excessive counter-threshold comparisons by using the aforementioned min-max heap property. The strategy involves examining max-level counters first since they have a higher count and are, therefore, more likely to be in $F$ (see Figure \ref{fig:QOSS_traversal_order}). A stack data structure facilitates traversal, along with a set $V$ to track visited parent and grandparent counters when traversing min-levels (lines 15 and 16). Traversal of the data structure begins at each of the children of the root counter, i.e., the counters with the largest count (line 17). If the count of a counter encountered on a max-level is greater than $N\phi$, it is output (lines 21-23), and traversal continues. Otherwise, all descendant counters are less than $N\phi$, and the traversal halts. Traversals reaching the bottom-most max-level continue by examining counters in min-levels upwards until the root counter is reached (lines 31-32 and 34-36). To avoid duplicate comparisons when performing an upward traversal from a (grand-)child to a (grand-)parent, each examined counter is added to a set $V$ of visited counters (lines 33 and 37).

\noindent\textbf{Query Time Complexity:} Using an underlying binary min-max heap, at most $5|F|$ counter-comparisons are performed during a query. This can be verified by selecting any subtree with a root node $r \in F$. Line 26 describes that no more than four counter-comparisons are performed. Performing this procedure for each $r \in F$ gives at most $|F| + 4|F| = 5|F|$ counter-comparisons. This approach can greatly reduce the number of counter-threshold comparisons when responding to a query in practice, mainly when the input follows a skewed distribution. 

The effects of the \shortsubalgorithmname improvements are studied in Section \ref{sec:evaluation}, comparing the algorithmic implementation to a baseline and evaluating the query latency and throughput.
A better query latency implies less overlapping with concurrent updates, facilitating improved accuracy. 

\begin{figure*}[htb!]
    \includegraphics[width=1\columnwidth]{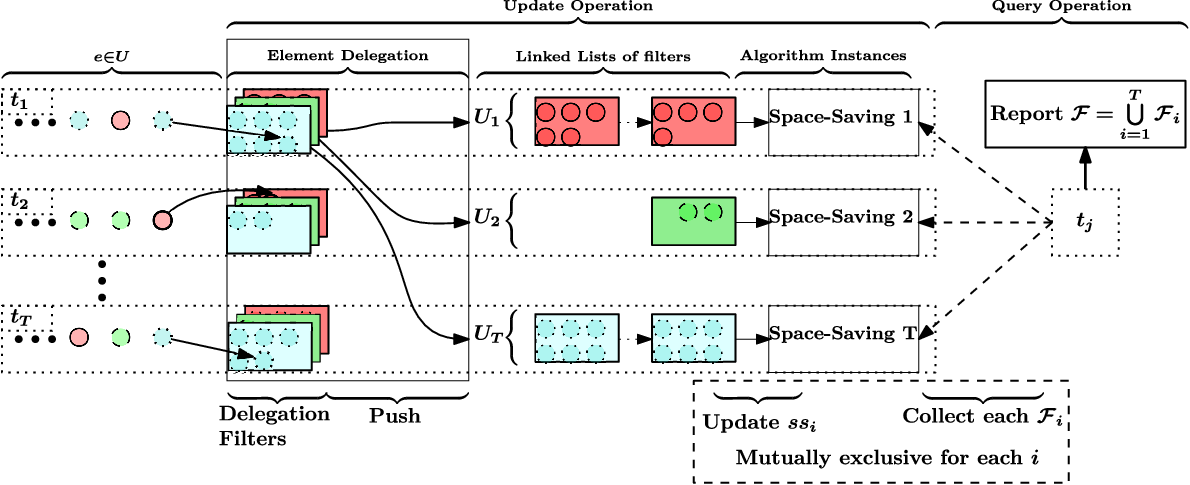}
    \caption{Overview of the update and query operations. Thread $t_1$ transfers full filters to the owner-threads for subsequent insertion into the reserved thread-local \shortsubalgorithmname data structures. Queries are mutually exclusive with insertions and gather the subset of frequent elements tracked by each thread into $F$.}
    \label{fig:overview_ins}
\end{figure*}

\begin{algorithm}[t!]
\small
\caption{Update operation on thread $j$. Constants $D$ and $E$ are user-defined and describe the filter size and the maximum number of processed elements per thread before handover.
}
\label{algo:update}
\begin{algorithmic}[1]
\Function{$UpdateQPOPSS$}{Element e}
\State $i \gets \textit{Owner }\text{(e)}$
\State $Filter \gets Threads[i].DelegationFilters[j]$ 
\State \newtext{(this $Filter$ is reserved for thread $j$)}
\If {$e \in \textit{Filter}$} 
    \State Increment count of e
\Else 
    \State Add e in Filter
    \State Set count of e to 1
\EndIf
\State Increment $N[j]$ by 1
\State Increment $c$ by 1
\State ($c$ counts updates since last handover)
\If {$Filter.size = D$ or $c = E$}
    \State Push all filters reserved for $j$ to respective owners linked list 
    \While {There are unflushed filters reserved for $j$ }
        \State $process\_pending\_updates()$ 
    \EndWhile
    \State $c \gets 0$
\EndIf
\EndFunction
\end{algorithmic}
\end{algorithm}

 \subsection{Concurrent Updates}
\label{sec:updates}
This section provides a detailed discussion of the update procedure in presence of concurrency, consisting of multiple algorithmic components arranged in a pipeline (see Figure \ref{fig:overview_ins}).
An arbitrary thread $j$ that processes an input-element $e$, owned by a thread $i=owner(e)$ (line 2, Algorithm \ref{algo:update}) inserts $e$ in the Delegation Filter of thread $i$ that is reserved for thread $j$ (lines 3-9).
 
 If an element is not already in the filter, it is added, with its count set to one. Otherwise, the count is incremented once. To keep track of the number of processed elements, both the counters $N[j]$
 (total filter insertions of the thread $j$) and $c$ (insertions since the last thread $j$ handed over its filter to the respective owners) are incremented once (lines 10 and 11). 
 The filters can be implemented as two arrays containing an element and count at the same index position (similar to Content-Addressable Memory \cite{pagiamtzis_content-addressable_2006}). 
 The length of the arrays is $D$, representing the number of distinct elements it can hold. If $D$ is kept small, the count of an element can be found efficiently by a simple linear search.

To enhance scalability and prevent filter staleness (the latter is elaborated in more detail in the next subsection),
elements are inserted in the appropriate Delegation Filters until any of these conditions is met:
a)~a filter becomes full by containing $D$ elements, or b)~$E$ elements have been processed by the thread since the previous handover. 
Either of these conditions (line 11) leads to the thread handing over all filters to their respective owners (line 12).
Handing over and flushing all filters periodically, especially in a skewed input distribution, enhances query accuracy since filter element counts are excluded from the query result (see Section \ref{sec:fequeries}).
A filter is handed over to the owner by pushing it to a multiple-producer single-consumer concurrent linked list reserved for each thread.
Thread $i$ continuously processes its own \emph{pending updates} until all its filters are marked empty (lines 15-16) before it resets $c$ to zero (line 17) and processes the next input.
 
 \algnewcommand{\LineComment}[1]{\State \(\triangleright\) #1}

\begin{algorithm}[tb!]
\small
\caption{Processing pending updates on thread $i$}
\label{algo:process_pending_updates}
\begin{algorithmic}[1]
\Function{$process\_pending\_update$}{}
\If {Threads[i].LinkedList empty}
    \State return
\EndIf
\If {Try-lock of Threads[i] taken
}
    \State return
\EndIf
\While {Threads[i].LinkedList is not empty}
    \State $Filter \gets Threads[i].LinkedList.pop()$ 
    \For{$each\ (Element\ e\ ,Weight\ w) \in\ Filter$}
        \State Threads[i].UpdateQOSS(e,w)
    \EndFor
    \State $Empty\ Filter$
    \State $Filter.size \gets 0$
\EndWhile
\State Threads[i].mutex $\gets 0$

\EndFunction
\end{algorithmic}
\end{algorithm}
 
\noindent\textbf{Processing pending updates:}
Each thread periodically checks for ready filters to process in its linked list, the absence of which causes immediate function termination (line 2-3 in algorithm \ref{algo:process_pending_updates}). Inversely, ready filters are processed given the successful acquisition of a thread-specific try-lock mutex (line 4-5), preventing possible data races due to concurrent update and query (see Sect.~\ref{sec:fequeries}) operations targeting the thread-local \shortsubalgorithmname instance. The lock is low in contention since it is only tested when the linked list contains a ready filter or when a query is carried out and can be implemented with a simple test-and-set variable. The elements of each filter are fed as weighted updates to the \shortsubalgorithmname instance (lines 8-9), and the filter is marked as empty (lines 10-11). Finally, the lock is released (line 12), allowing a querying thread to read the \shortsubalgorithmname data structure.

\subsection{Concurrent Frequent Elements Queries}
\label{sec:fequeries}
\begin{algorithm}[b!]
\small
\caption{Query operation on thread $j$}\label{pseudocode:query}
\label{algo:frequent_elements_query}
\begin{algorithmic}[1]
\Function{$QueryQPOPSS$}{$\phi\in [0,1]$}
\State $N \gets sum(N_i), i \in \{1..T\}$
\While{there exists QOSS$_i$, $i \in \{1..T\}$ not yet queried}
    \If{ Try-lock of Threads[i] taken 
    }
        \State Try another thread
    \Else
        \State Threads[i].QueryQOSS($\phi$,$N$)
        \State Release try-lock of Threads[i]
    \EndIf
    \State process\_pending\_updates()
\EndWhile
\State Output frequent elements

\EndFunction
\end{algorithmic}
\end{algorithm}

Any thread out of the $T$ dispatched ones may answer a frequent elements query while other threads are updating or querying (i.e., concurrently), the implications of which are discussed in Sect. \ref{sec:algo_impl_analysis}. A query aims to report the set of frequent elements from definition \ref{def:epsapprox}.
To calculate the threshold value $N\phi$, needed to select the frequent elements efficiently, the query procedure begins by calculating an estimate of the stream length, $N$. This value is computed as the sum of the elements processed by each thread, $N[i]$ (line 2 in Algorithm \ref{algo:frequent_elements_query}). To collect the subset of frequent elements tracked by each thread-local \shortsubalgorithmname algorithm, the querying thread tries to acquire the test-and-set-lock associated with each thread (line 3). Once the lock of a thread has been acquired, a query is issued to the corresponding \shortsubalgorithmname algorithm (line 7).
 
Recall that the try-lock acquisition is a non-blocking action; if a thread cannot acquire it immediately, it will simply retry later. Meanwhile, the thread can do other useful work, such as processing its pending updates (line 9).
Due to the aggregation of elements in Delegation Filters, the contention on these accesses is low.
Furthermore, in the \shortsubalgorithmname algorithmic implementation (section~\ref{sec:QOSS}), a query is processed in $O(|F|)$ time, where $|F|$ is the number of frequent elements, reducing the contention further.

\noindent
\textbf{Query scalability enhancement:}
We identify a performance tradeoff in the design of Delegation Filters: During a query, buffered element occurrences in the Delegation Filters can be ignored to improve query speed and overall throughput. However, this introduces a slack between when an element is first observed and subsequently reported in a query. By bounding the maximum sum of element counts in a Delegation Filter by a constant $E$, we can guarantee a fixed handover delay, which can be kept low concerning the usual inaccuracy inherent to data synopses.
At the same time, we want to maximize the portion of elements inside \shortsubalgorithmname data structures. We achieve this by introducing the aforementioned mechanisms concerning $D$ and $E$ for promptly handing over filters to owner threads at a fixed rate.
Indeed, initial experiments suggested that exploring this tradeoff yielded up to 1.73x higher throughput and 0.5x lower query latency (with 24 threads, $\phi=10\epsilon = 0.0001$, E=1000, D=32, querying on average once per 10 updates using real IP-packet data), compared to the non-enhanced approach, while bounding the reporting delay to concern less than $E$ element occurrences times the number of threads. 
The side effect on the approximate output introduced by the enhancement diminishes rapidly with the length of the execution, as we show in Section~\ref{sec:algo_impl_analysis}.

\section{\newtext{ Analysis}}
\label{sec:algo_impl_analysis}
Having described our method for estimating the frequent elements, we now focus on the space requirements and estimation guarantees under concurrent updates and queries. To aid us in this task, we define a set of symbols common to our analysis in table \ref{table:symbols}.

We initiate the discussion using a meta-lemma containing useful lemmas and theorems from \cite{metwally_efficient_2005}, that apply to the \shortsubalgorithmname algorithm.
\begin{lemma} \shortsubalgorithmname preserves the following properties (implied from the respective lemmas and theorems in ~\cite{metwally_efficient_2005})
    \label{claim:meta_imported}
    \begin{enumerate} [topsep=0pt, leftmargin=*]
    \setlength\itemsep{0em}
        \item (From Lemma 3.3 in \cite{metwally_efficient_2005}) If the number of counters $m$ is chosen such that $m = \frac{1}{\epsilon}$, then the minimum counter value of \shortsubalgorithmname, denoted as $ F_{\min}$, is less than or equal to $\lfloor N \epsilon \rfloor$, where $N$ is the length of the stream.
        \item (From Theorem 3.5 in \cite{metwally_efficient_2005}) Any element that occurs more than $F_{min}$ times in $\mathcal{S}$ is guaranteed to be tracked by \shortsubalgorithmname.
        \item (From Lemma 3.4 in \cite{metwally_efficient_2005}) 
        Elements tracked by \shortsubalgorithmname are overestimated by at most $F_{min}$. In other words: $f(e) \leq \hat{f}(e)\leq f(e) + \epsilon N$.
        \item (From Lemma 4.3 in \cite{metwally_efficient_2005}) 
        Any element that occurs in $\mathcal{S}$ more frequently than the maximum possible value of $F_{\min}$ is guaranteed to be reported, regardless of stream order.
    \end{enumerate}
\end{lemma}

\begin{table}[hb!]
\centering
\begin{tabular}{ | c | c | }
   \hline
   \emph{Symbol}    &    \emph{Description}\\
   \hline
      $\epsilon$ & User-specified approximation factor.\\
   \hline
    $\phi$ & User-specified support threshold.\\
    \hline
   $\mathcal{S}$ & The stream of elements.\\
   \hline
   $N$ & The stream length of $\mathcal{S}$. 
   \\
   \hline 
   $U$ & Domain of $\mathcal{S}$.\\
   \hline
   $e$ & An element of a stream.\\
   \hline
   $r(e)$ & The rank of a stream element.\\
   \hline
      $f_N(e)$ & The number of occurrences of $e$ in $\mathcal{S}$.\\
   \hline
    $\hat{f}_N(e)$ & The estimated number of occurrences of $e$ in $\mathcal{S}$.\\
    \hline
   $m$ & The number of counters in \shortsubalgorithmname.\\
   \hline
    $F$ & Set of elements and estimated occurrence in $\mathcal{S}$ tracked by \shortsubalgorithmname.\\
   \hline
   $F_{min}$ & Least estimated occurrence of an element in $\mathcal{S}$ tracked by \shortsubalgorithmname.\\
    \hline
    $T$ & Number of dispatched threads.\\
   \hline
   $\mathbb{D}(e)$ & The number of counts of $e$ in a Delegation Filter.\\
   \hline
    $E$ & Parameter controlling the number of elements in delegation filters.\\
   \hline
    $D$ & Number of slots in a delegation filter.\\
   \hline
   $\zeta$ & The Euler–Riemann function.\\
   \hline
   $a$ & Skew parameter for Zipf distribution.\\
 \hline
    $N_S$ & Stream length at the start of a query.\\
 \hline
    $N_E$ & Stream length at the end of a query.\\
 \hline
\end{tabular}
\caption{Descriptions of symbols used.}
\label{table:symbols}
\end{table}

The rest of this section adheres to the following structure: 
First, we show that maintaining accuracy requires fewer counters when the \shortsubalgorithmname algorithm observes a stream of elements belonging to a subset of the original domain of possible elements. Second, we describe the counter requirements of \shortalgorithmname, composed of multiple \shortsubalgorithmname.
Lastly, we focus on the consistency guarantees of the frequent elements in the face of concurrently overlapping queries and updates and provide consistency-implied accuracy bounds.

\subsection{Domain Splitting and Space Requirements}
We begin by analyzing the number of counters required by \shortsubalgorithmname to accurately report the frequent elements defined in Section \ref{sec:preliminaries} when processing elements from a split-domain stream, i.e. a stream where elements not in a specific subset of the universe of possible elements are omitted.

To this end, we introduce $\mathcal{S}_{1..x}$ as the $x$-prefix of an unbounded stream $\mathcal{S}$, containing the first $x$ elements.
Each symbol in the sequence, called $\mathcal{S}_i$ for each $i \in \{1..x\}$, can be found in $U \cup \{\varnothing\}$, the universe of possible elements in union with the null symbol. We include $\varnothing$ to denote the absence of an element, used for highlighting element-wise differences between streams. Furthermore, we introduce a function to transform a stream to a \emph{stream block} (analogous to a set block \cite{brualdi_introductory_2010}), containing only stream elements from a specific set block $B$ of a partition of $U$. What follows are three function definitions for constructing a stream block, counting the number of deleted elements in a stream, and finding the length of a bounded stream.
\[\text{block}(\mathcal{S},B) = \begin{cases} 
    \mathcal{S}_1\ \text{if}\ \mathcal{S}_1\in B\ \text{else}\ \varnothing & \text{if\ } x = 1 \\
     (\varnothing,\text{block}(\mathcal{S}_{2...x},B))    & \text{if\ } \mathcal{S}_1 \notin B  \\
    (\mathcal{S}_1,\text{block}(\mathcal{S}_{2...x},B)) & \text{if\ } \mathcal{S}_1 \in B \\
   \end{cases}
\]
\[\text{\#del}(\mathcal{S}) = \begin{cases} 
      1\ \text{if}\ \mathcal{S}_1 = \varnothing\ \text{else}\ 0 & \text{if\ } x = 1  \\
      \text{\#del}(\mathcal{S}_{2...x}) & \text{if\ } \mathcal{S}_1 \neq \varnothing \\
     1 + \text{\#del}(\mathcal{S}_{2...x})    & \text{if\ } \mathcal{S}_1 = \varnothing 
   \end{cases}
\]
\[\text{len}(\mathcal{S}) = \begin{cases} 
      0\ \text{if}\ \mathcal{S}_1 = \varnothing\ \text{else}\ 1 & \text{if\ } x = 1  \\
      1 + \text{len}(\mathcal{S}_{2...x}) & \text{if\ } \mathcal{S}_1 \neq \varnothing \\
     \text{len}(\mathcal{S}_{2...x})    & \text{if\ } \mathcal{S}_1 = \varnothing 
   \end{cases}
\]
Using these definitions, a stream block can be constructed as in the following example: if $U=\{a,b,c,d\}$, $\mathcal{S}_{1..5} \coloneqq a,a,b,d,c$, and $B = \{a,c\}$, then $\mathcal{B}_{1..5} \coloneqq$ $block(\mathcal{S},B) = a,a,\varnothing,\varnothing,c$,
\#del($\mathcal{B}$)=2, and len($\mathcal{B}$) $=3$.

From this point, intending to bound the number of counters needed by \shortalgorithmname to produce the frequent elements from Definition \ref{def:epsapprox}, we observe the relation between a stream's domain size and the minimum counter in \shortsubalgorithmname.
\begin{lemma}
\label{claim:block_counters}
When \shortsubalgorithmname observes the stream block $\mathcal{B} \coloneqq $ block$(\mathcal{S},B)$ and maintains $m=\frac{1}{T\epsilon}$ counters, the minimum counter is at most $\lfloor\frac{N}{\epsilon}\rfloor$,
if $|B| = \lceil \frac{|U|}{T} \rceil$.
\end{lemma}
\begin{proof}
 Let $j$ and $L$ be arbitrary positive integers. Consider a stream $\mathcal{S}$ with length $len(\mathcal{S}) = L(m+1+j)$, containing $m+1+j$ distinct elements, each repeated $L$ times. Let these elements belong to the set $U$, such that $|U| = m+1+j$.
Given $|B| = \lceil \frac{|U|}{T} \rceil$, the stream $\mathcal{B} \coloneqq block(\mathcal{S},B)$ has a length of $len(\mathcal{B}) = L\frac{m+1+j}{T}$.
We utilize claim 1 in Lemma \ref{claim:meta_imported} to determine the minimum counter of \shortsubalgorithmname while observing $\mathcal{B}$:
$$F_{min} \leq
\lfloor L \frac{m+1+j - \frac{(m+1+j)(T-1)}{T}}{m}\rfloor = \lfloor \epsilon L \big(m + 1 + j\big)  \rfloor = 
\lfloor N\epsilon\rfloor $$
\end{proof}

This follows naturally from the linear relationship between consumed space and the accuracy of the Space-Saving algorithm and has implications for the required number of counters of QOSS under domain splitting in the following.
\begin{lemma}
\label{claim:counters_general}
When \shortsubalgorithmname observes the stream block $\mathcal{B} \coloneqq $ block$(\mathcal{S},B)$ and maintains $m = \frac{1}{T\epsilon}$ counters, it tracks every element occurring more than $\lfloor N\epsilon \rfloor$ times with an estimation error of at most $\lfloor N\epsilon \rfloor$ if $|B| = \lceil \frac{|U|}{T} \rceil$.
\end{lemma}
\begin{proof}
From Lemma \ref{claim:block_counters}, we establish that \shortsubalgorithmname maintains a minimum counter value $F_{\min} \leq \lfloor N\epsilon \rfloor$. Leveraging claims 2 and 3 in Lemma~\ref{claim:meta_imported}, elements occurring more than $\lfloor N\epsilon \rfloor$ times are guaranteed to be tracked, with an overestimation error of at most $\lfloor N\epsilon \rfloor$.
\end{proof}

Since \shortalgorithmname consists of multiple \shortsubalgorithmname instances, the two's space requirements and accuracy guarantees are interlinked. By its design, \shortalgorithmname dispatches $T$ \shortsubalgorithmname algorithmic instances, one for each thread. According to Lemma \ref{claim:counters_general}, each requires $\frac{1}{T\epsilon}$ counters to track the frequent elements. Therefore, the total number of counters needed to track and report the $\epsilon$-approximate frequent elements of $\mathcal{S}$ are $T\frac{1}{T\epsilon} = \frac{1}{\epsilon}$.
\begin{corollary}
\label{claim:qpopss_space_requirement}
\shortalgorithmname{} requires $\frac{1}{\epsilon}$ counters to track every element in the stream $\mathcal{S}$ occuring more than $\lfloor N\epsilon \rfloor$ times, with an estimation error of at most $\lfloor N\epsilon \rfloor$.
\end{corollary}
We now investigate the number of required counters, assuming that the input data stream conforms to the Zipf distribution.
\begin{theorem}
\label{claim:counters_zipf}
\shortsubalgorithmname{} with $m = \left(\frac{1}{\epsilon T}\right)^{\frac{1}{a}}$ counters, fed with stream $\mathcal{B} \coloneqq \text{block}(\mathcal{S}, B)$, tracks every element occurring more than $N\phi$ times and reports occurrences with an error of at most $\lfloor N\epsilon \rfloor$, provided $\mathcal{S}$ is constructed from a noiseless Zipf distribution with $a > 1$, regardless of stream permutation.
\end{theorem}
\begin{proof}
According to claim 4 in \ref{claim:meta_imported}, \shortsubalgorithmname reports frequent elements occurring more often than the maximum possible value of $F_{min}$. For a Zipf-distributed input stream, the maximum value of $F_{min}$ is less than or equal to the cumulative occurrences of the elements ranked between $m+1$ and $|U|$, divided equally over the number of counters: $F^{zipf}_{min} \leq \frac{N}{m} \frac{\sum_{i=m+1}^{|U|} \frac{1}{i^a}}{\sum_{i=1}^{|U|} \frac{1}{i^a}}$. Similarly, the number of occurrences of an element of a particular rank is described as $\frac{N}{r(e)^a} \frac{1}{\sum_{i=1}^{|U|} \frac{1}{i^a}}$.  The following proof obligation describes the ranks of elements whose occurrences exceed the maximum value of $F^{zipf}_{min}$:
\begin{equation*}
\label{eq:obligation}
   \frac{1}{r(e)^a}>\frac{1}{m} \sum_{i=m+1}^{|U|} \frac{1}{i^a} 
\end{equation*}

The right-hand side of inequality can be simplified:
\begin{equation*}
    \frac{1}{r(e)^a}>\frac{1}{m^a} \sum_{i=2}^{\frac{|U|}{m}} \frac{1}{i^a}
\end{equation*}
Since $\sum_{i=2}^{\frac{|U|}{m}} \frac{1}{i^a}$ has no closed-form expression, $\zeta(a)-1$ is used as a substitute, imposing a greater constraint on $m$:
\begin{equation*}
    \frac{1}{r(e)^a}>\frac{1}{m^a} (\zeta(a)-1)
\end{equation*}
Given that $B$ is formed by randomly selecting $\frac{|U|}{T}$ elements from $U$ with uniformity (approximately), the cumulative elements denoted by $(\zeta(a)-1)$ can be assumed to be evenly distributed among threads
\footnote{This assumption is motivated by the tendency for the total occurrences of the $|U| - (m+1)$ least frequent elements to be low and for $|U| - (m+1)$ to be much larger than $m$.}:
\begin{equation*}
\label{eq:introduceT}
    \frac{1}{r(e)^a}>\frac{1}{m^a} \frac{(\zeta(a)-1)}{T}
\end{equation*}
This simplifies to:
\begin{equation}
\label{eq:a5}
    m > r(e) \bigg(\frac{\zeta(a)-1}{T}\bigg)^\frac{1}{a}
\end{equation}
Now, considering the inequality $\frac{N}{r(e)^a \zeta(a)} < N\epsilon$, satisfied for element ranks that occur more often than the threshold. The inequality can be solved for $r(e)$ to obtain the least element rank satisfying the above inequality:
\begin{equation*}
    r(e)\geq\bigg(\frac{1}{\epsilon\zeta(a)}\bigg)^\frac{1}{a}
\end{equation*}
We can now substitute $r(e)$ in inequality \ref{eq:a5}: 
\begin{equation*}
    m > \bigg(\frac{\zeta(a)-1}{T\epsilon \zeta(a)}\bigg)^\frac{1}{a}= \bigg(\frac{1}{T\epsilon} - \frac{1}{T\epsilon\zeta(a)}\bigg)^\frac{1}{a}
\end{equation*}
Thus, setting $m=\big(\frac{1}{T\epsilon}\big)^\frac{1}{a}$ guarantees that the proof obligation in (\ref{eq:obligation}) is satisfied, i.e., elements with rank less than $\bigg(\frac{1}{\epsilon\zeta(a)}\bigg)^\frac{1}{a}$ exceed the maximum minimum counter value, and will therefore be reported by \shortsubalgorithmname.
\end{proof}
Having determined the space requirements of \shortsubalgorithmname when observing a stream with Zipfian distribution, we can discuss the space requirements of \shortalgorithmname. \shortalgorithmname dispatches $T$ threads, each with its own \shortsubalgorithmname algorithm instance requiring $\big(\frac{1}{T\epsilon}\big)^\frac{1}{a}$ counters according to  Theorem \ref{claim:counters_zipf}.

\begin{corollary}
\label{claim:qpopss_zipf_requirements}
\shortalgorithmname requires $T\big(\frac{1}{T\epsilon}\big)^\frac{1}{a}$ counters to track every element of a stream $\mathcal{S}$ occuring more than $\lfloor N\epsilon \rfloor$ times and reports the number of occurrences of elements with an error of at most $\lfloor N\epsilon \rfloor$, given that $\mathcal{S}$ is constructed from a noiseless Zipf distribution with $a>1$, regardless of stream permutation.
\end{corollary}

Ultimately, corollaries \ref{claim:qpopss_space_requirement} and \ref{claim:qpopss_zipf_requirements} describe the required number of counters needed by \shortalgorithmname to track the $\epsilon$-approximate $\phi$-frequent elements across various distributions.
Additionally, the Delegation Filters described in \ref{sec:updates} contain counters equal to the number of threads squared times the number of counters kept by each filter ($T^2D$).
Having established space requirements, we now explore the query consistency guarantees of \shortalgorithmname.

\subsection{Query Consistency Guarantees}
\label{sec:consistency_guarantees}
In this section, we provide an invariant for the approximation guarantees of the frequent elements and their occurrence reported by  \shortalgorithmname as they relate to challenge \ref{challenge:accuracy_bound}.
We use a similar reasoning and method as Rinberg and Keidar~\cite{RinbergKeidarDisc2020}, who defined bounds for the estimated count of an element on a concurrent Count-Min Sketch.
The authors utilize that counters of a Count-Min Sketch are monotonically increasing, as is the case with counters in \shortsubalgorithmname, which is used as a basis for a method to bound the error for a query in this work.
Before we delve into the consistency analysis, we first discuss the implications of the algorithm design.

\noindent
\textbf{Query Scalability Enhancement:}
As implied by the \emph{query scalability enhancement} of \shortalgorithmname described in section~\ref{sec:updates}, the parameter $E$ expresses the maximum number of elements present in a delegation filter at any point in time. Since there are $T$ delegation filters that can contain an element, the maximum number of element occurrences that can be missing from a reported element count is $T\cdot E$, which is put more concisely as the following Lemma.
\begin{lemma}
\label{claim:underestimation}
When $\mathcal{S}$ consists of elements drawn from an arbitrary distribution, $\mathbb{D}(e)$, the number of counts of $e$ inside delegation filters, is less than $T \cdot E$.
\end{lemma}
 Note that the Lemma is a rather substantial overestimation. In common executions, Delegation Filters contain various elements. Being fully occupied by a single element is rather unlikely for each Delegation Filter. 

Suppose the input distribution is noiseless Zipf with skew parameter $a>1$. In that case, we can give a tighter bound on the number of counts of a particular element $e$ inside Delegation Filters.

\begin{lemma}
\label{claim:abs_err}
When $\mathcal{S}$ 
consists of elements drawn from a noiseless Zipf distribution with infinite domain and skew parameter $a>1$, $\mathbb{D}(e)$ is at most $\frac{T \cdot E}{\zeta(a)r(e)^a}$.
\end{lemma}

\noindent
\textbf{Query consistency:}
To capture the notion of queries that are concurrent with updates, we introduce $N_S$ and $N_E$, which describe the stream length at the start and end of a query. These values are ordered such that $N_S \leq N \leq N_E$. We update claim 3 in Lemma \ref{claim:meta_imported} to capture potential concurrent update operations during a query in the following.

\begin{lemma}
\label{claim:consistency_count}
Given a stream $\mathcal{S}$ of $N$ elements, \shortalgorithmname estimates the occurrence of an element $e \in U$ such that $f_{N_{S}}(e) - \mathbb{D}(e) \leq \hat{f_N}(e)\leq f_{N_E}(e) + \epsilon N_{E}$.
\end{lemma}
\begin{proof}
We have that $f_{N_S}(e) \leq f_{N_E}(e)$.
There are at most $\mathbb{D}(e)$
counts of element $e$ that have not yet been inserted into the \shortsubalgorithmname instance of the owner of $e$, therefore, $f_{N_S}(e) - \mathbb{D}(e) \leq \hat{f_N}(e)$ is the minimum value a counter can assume. The maximum over-estimation of \shortsubalgorithmname is $\epsilon N$, which is maximized at the end of a query when $N_{E}$ elements have been processed. Therefore, the estimated count of an element is at most $\hat{f_N}(e) \leq f_{N_E}(e) + \epsilon N_{E}$.
\end{proof}

We provide a consistency guarantee for the set of elements reported by \shortalgorithmname, aligning with Definition \ref{def:epsapprox} for queries spanning more than 0 updates.

\begin{theorem}
\label{claim:consistency_freq_elems}
 Given a stream $\mathcal{S}$ of $N$ elements, 
\shortalgorithmname is guaranteed to report the set $F$ containing all elements $e\in U$ with $f(e)_{N_S} > \phi N_{S} + \mathbb{D}(e)$, and no elements $e\in U$ with $f(e)_{N_E} < (\phi-\epsilon) N_{S} - \epsilon(N_{E}-N_S)$, where $0<\epsilon\leq \phi<1$.
\end{theorem}

\begin{proof}
The \shortalgorithmname algorithm reports all elements with an estimated count 
\begin{equation}
\label{eq:returned}
\hat{f(e)}_{N_S} > \phi N_S
\end{equation}
This is true since $N_S$ is calculated at the start of a query, all elements with $\hat{f_N}(e)>\phi N_S$ are reported, and \shortsubalgorithmname counters increase monotonically.

As shown in Lemma \ref{claim:consistency_count}, $\hat{f(e)}$ is at least (a) $f(e)_{N_S} - \mathbb{D}(e)$ and at most (b) $f(e)_{N_E} + N_E \epsilon$. 
Substituting $\hat{f(e)}$ for (a) in (\ref{eq:returned}) yields:
\begin{equation}
\label{eq:consist_recall}
    f(e)_{N_S} > \phi N_S+\mathbb{D}(e)
\end{equation}
Substituting $\hat{f(e)}$ for (b) in (\ref{eq:returned}) yields:
\begin{equation*}
    \label{eq:consistency_max_error}
    f(e)_{N_E} > \phi N_S- \epsilon N_E
\end{equation*}
Expression (a) and the negation of (b) together give that all elements $f(e)_{N_S} > \phi N_S$ and no elements $f(e)_{N_E} < \phi N_S- \epsilon N_E$ are reported by \shortalgorithmname.
\end{proof}

According to Theorem \ref{claim:consistency_freq_elems} \shortalgorithmname will report all elements occurring more than $\phi N_S$ times after processing $N_S$ elements, given that delegation filters are empty. With fixed-size delegation filters, \shortalgorithmname tends to report all elements more frequently than $\phi N_S$ as the stream length grows in relation to the delegation filter size. This notion is formalized in Theorem \ref{claim:perfectrecall_general} below.
\begin{theorem}
\label{claim:perfectrecall_general}
As the length of the stream tends to infinity, \shortalgorithmname reports \emph{all}
$\epsilon$-approximate frequent elements.
\end{theorem}

\begin{proof}
We start by normalizing inequality (\ref{eq:consist_recall}):
\begin{equation*}
 \frac{f(e)_{N_S}}{N_S} > \phi + \frac{ \mathbb{D}(e)}{N_S }
\end{equation*}
The element $e$ occurs a fraction $P(e)= \frac{f_{N_S}(e)}{N_S}$ of the time in the stream. We then have that:
\begin{equation*}
     \lim_{N_S \to \infty} P(e) > \phi + \frac{ \mathbb{D}(e)}{N_S } \to P(e) > \phi
\end{equation*}
Thus, when enough elements have been processed, all elements with $P(e) > \phi$ are reported.
\end{proof}

\begin{theorem}
\label{claim:perfectrecall_zipf}
As $N$ tends to infinity, \shortalgorithmname achieves perfect recall when processing a stream constructed by drawing elements from a Zipf distribution with parameter $a>1$.
\end{theorem}

\begin{proof}
Using the expression from inequality (\ref{eq:consist_recall}), we can rewrite it in its Zipfian form:
\begin{equation*}
\frac{1}{\zeta(a)r(e)^a} > \phi + \frac{\mathbb{D}(e)}{\zeta(a)r(e)^a N_S}
\end{equation*}
Solving for r(e) gives:
\begin{equation*}
    \bigg( \frac{1}{\zeta(a)\phi} - \frac{\mathbb{D}(e)}{\zeta(a)\phi N_S}\bigg)^\frac{1}{a} > r(e)
\end{equation*}
Simplifying yields:
\begin{equation*}
\label{eq:zipf_perfectrecall}
    \bigg( \frac{1}{\zeta(a)\phi} \bigg)^\frac{1}{a} \bigg(1 - \frac{\mathbb{D}(e)}{N_S}\bigg)^\frac{1}{a} > r(e)
\end{equation*}
We then have that:
\begin{equation*}
    \lim_{N_S \to \infty} \bigg( \frac{1}{\zeta(a)\phi} \bigg)^\frac{1}{a} \bigg(1 - \frac{\mathbb{D}(e)}{N_S}\bigg)^\frac{1}{a} > r(e) \to \bigg( \frac{1}{\zeta(a)\phi} \bigg)^\frac{1}{a} > r(e)
\end{equation*}
Meaning that all elements with a rank lower than $\bigg( \frac{1}{\zeta(a)\phi} \bigg)^\frac{1}{a}$ are guaranteed to be reported given that enough elements have been processed.
\end{proof}
To summarize the results, this analysis underscores the memory-space-related benefits of operating on a subset of the original domain (Corollaries 
\ref{claim:qpopss_space_requirement} and \ref{claim:qpopss_zipf_requirements}). It also highlights the impact of excluding elements within Delegation Filters on query consistency and accuracy, offering bounds for the latter (Lemmas \ref{claim:abs_err} and \ref{claim:underestimation}). Moreover, it addresses the consistency guarantees of frequent elements amidst concurrent queries and updates (Theorems \ref{claim:consistency_freq_elems}, \ref{claim:perfectrecall_general}, and \ref{claim:perfectrecall_zipf}). 
Having evaluated the analytical properties of \shortalgorithmname, we continue 
with an experimental evaluation.

\section{Evaluation}\label{sec:evaluation}

This section contains the in-depth empirical evaluation of \shortalgorithmname 's properties.
We begin by describing the experiment setup in detail; then, we investigate the performance of \shortsubalgorithmname compared to Space-Saving. We also compare \shortalgorithmname and the representative works \cite{zhang_efficient_2014, TopKapi} on throughput and scalability, accuracy and memory requirements, as well as query latency.

\subsection{Experimental Setup}
\label{sec:experimental_setup}

\noindent\textbf{Computing platform:} 
The experiments were carried out on a server with 2 Intel Xeon X5675 processors, each with 6 physical cores and 2-way hyperthreading enabled, for 12 physical and 24 logical cores. Each core operates at a clock speed of 3.07GHz, with cache sizes L1 32KB, L2 256KB, and L3 12MB. The main memory was 70 GB, and the installed operating system was Debian 10.9. The code was compiled for the x86 architecture using GCC 8.3.

\noindent\textbf{Data sets:}
Both synthetic and real data were used in the evaluation of \shortalgorithmname.
Unless otherwise stated, the synthetic data sets contain $100M$ elements, sampled from a universe of $|U|=100M$ elements according to the probability mass function of the Zipf distribution such that $f_N(e)=\frac{N}{H_{|U|,a}r(e)^a}$~\cite{zipf1949human}, where $H_{|U|,a} = \sum_{i=1}^{|U|}\frac{1}{i^a}$. 
In total, 11 synthetic data sets were created from Zipf distributions with skew parameter $a$ ranging from $0.5$ to $3$ in increments of $0.25$. Hereinafter, we use the term skew to mean $a$ in $f_N(e)=\frac{N}{H_{|U|,a}r(e)^a}$.

\begin{table}[ht]
  \begin{varwidth}[b]{0.45\linewidth}
    \centering
\begin{tabular}{|l|l|l|l|}
\hline
\diaghead{\hskip1.5cm}{Data set}{$\phi$}        & $10^{-3}$ & $10^{-4}$ & $10^{-5}$ \\ \hline
CAIDA & 44    &   1555     & 10463   \\ \hline
Zipf a=1.25 & 74 & 467 & 2952 \\
\hline
Zipf a=2 & 24 & 77 & 246 \\
\hline
Zipf a=3 & 9 & 20 & 43 \\
\hline
 \end{tabular}
    \caption{The number of frequent elements for different threshold values of $\phi$ in the CAIDA and selected Zipf data sets.}
    \label{table:frequent_elems}
  \end{varwidth}%
  \hfill
  \begin{minipage}[b]{0.50\linewidth}
    \centering
    \vspace{0.5em}
    \includegraphics[width=0.95\columnwidth]{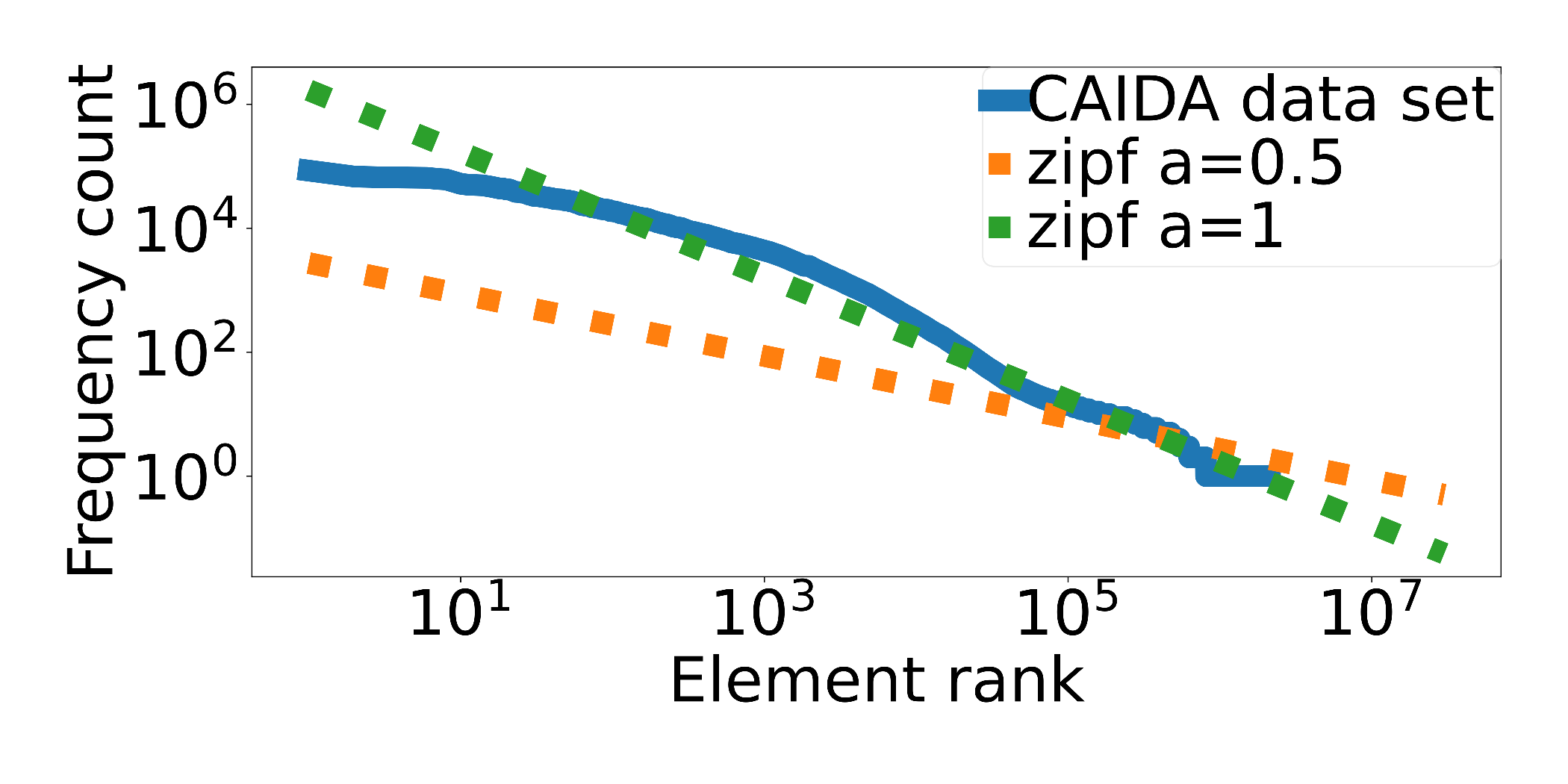}
    \captionof{figure}{Rank and count of each unique element in the CAIDA data set. Zipf distributions with skew 0.5 and 1 are plotted as a guide. Note the logarithmic scale on x- and y-axes.}
    \label{fig:real_data_distribution}
  \end{minipage}
\end{table}

A real-world data set was extracted from the CAIDA Anonymized Internet Traces 2019 data set~\cite{CAIDA} by selecting an arbitrary 60-minute window of IP packet traffic in an arbitrary direction of a backbone interface. Each packet contains a 5-tuple (flow) of source-destination IP addresses, source-destination ports, and protocol used. The set contains roughly $21M$ packets belonging to around $2.1M$ unique flows. As shown in figure~\ref{fig:real_data_distribution}, the distribution of the flows is similar to that of synthetic data generated with a skew parameter of 1.
The number of frequent elements in the data set is detailed in table~\ref{table:frequent_elems} for different values of $\phi$. As for the synthetic data, the expression $\big(\frac{1}{\zeta(a)\phi}\big)^\frac{1}{a}$ describes the least element rank for a certain threshold value $\phi$ and Zipf distribution skew parameter~$a$.

\noindent\textbf{Metrics:} 
Evaluation metrics include query and update \emph{throughput} (millions of operations per second), \emph{accuracy}, \emph{memory consumption} (megabytes reserved), and the \emph{latency} (the time between the start and end of a query in microseconds). 
More specifically, accuracy is measured as \emph{average relative error} (ratio of estimated element occurrences to actual occurrences), \emph{precision} (ratio of actual positive frequent elements to the number of reported frequent elements), and \emph{recall} (ratio of actual positive frequent elements to the number of actual frequent elements). 

\noindent
\textbf{Measurement Methodology:}
Our experiments measured throughput as the number of operations (updates and queries) per time unit. Throughput experiments were executed for 10 seconds while processing a stream of 100 Million elements repeatedly for the duration.
The number of counters of each baseline was set according to the respective theoretical bound. The query latency was calculated by measuring processor clock cycles between a query's start and end. This was done for multiple queries whose mean value was calculated. The accuracy metrics were measured by processing a stream and issuing a single query at the end. The memory consumption was calculated by selecting an accuracy level and computing the memory consumption for the different baselines according to their theoretical space/accuracy bounds.

\medskip

\noindent\textbf{Baselines:}
Our  of \shortalgorithmname (which we share with the community in open-source format in~\cite{Delegation_Space-Saving_2022}, along with the code to generate the synthetic data and links to the real data used in this evaluation) was compared to the following baselines (motivated also in Section~\ref{sec:problem_analysis}, with more detail provided here, so as to better explain the associations and comparisons)
\begin{enumerate}[topsep=0pt, leftmargin=*]
\setlength\itemsep{0em}
    \item A single-threaded \shortsubalgorithmname 
    \item \shortalgorithmname using Space-Saving as the inner algorithm
    \item PRIF \cite{zhang_efficient_2014}
    \item Topkapi~\cite{TopKapi}
\end{enumerate} 
To examine the speedup between a single-threaded \shortsubalgorithmname and the multi-threaded \shortalgorithmname,  (1) was selected as a baseline, while (2) facilitates studying the impact of the query timeliness-improvement of \shortsubalgorithmname.
(3) and (4) were chosen since they are representative multi-threaded approaches to the frequent elements problem.

\textit{PRIF}~\cite{zhang_efficient_2014} entails a reserved merging thread that periodically merges updates from thread-local algorithm instances. 
As an algorithmic component, the authors present OWFrequent, an optimized version of Frequent~\cite{misra_finding_1982} that supports weighted updates. Due to the merging thread, extra latencies are introduced. The authors propose an update coefficient $\beta$ that controls the rate at which the merging thread receives updates from the thread-local OWFrequent algorithm instances.
The authors give a rigorous analysis and evaluation of the approach. The evaluation shows that PRIF is somewhat precise in reporting the frequent elements and that there is a good speedup compared to a single-threaded version.
Due to the absence of open-source implementations of PRIF, one was created for evaluation purposes \cite{Delegation_Space-Saving_2022}. 
As in the authors' implementation, a shared-bounded buffer was implemented with semaphores \cite{robbins_unix_2003} to handle the communication between the sub-threads and the merging thread. 
OWFrequent was implemented using the frequent items sketch algorithm package 
of~\cite{finding_frequent_items_website} as a base. 
The implementation of \shortsubalgorithmname was also based on said algorithm package.

The \emph{Topkapi Sketch}~\cite{TopKapi} combines the concepts of the Frequent Algorithm with the Count-Min Sketch by keeping a Frequent counter in each cell of a Count-Min Sketch matrix. 
A Topkapi Sketch with $\log(2\frac{N}{\delta})$ rows and $\frac{1}{\epsilon}$ counters (see analysis section in \cite{cormode_improved_2005}), solves the $\phi$-approximate frequent elements problem outlined in definition \ref{def:epsapprox}.  The authors present a multi-threaded approach wherein multiple streams are processed concurrently. At the end of processing, the summaries are merged into a final result, promising to deliver the frequent elements of the combined stream, adhering to the theoretical limits outlined. The code is open source and was adopted for the relative study in this paper.

\medskip

\noindent\textbf{Baseline adaptations:}
To ensure a fair comparison, certain adaptations were made.
Since PRIF supports concurrent queries and updates, we implemented this algorithm \emph{as is}.                                
Regarding Topkapi, as it lacks design elements that enable concurrent queries and updates, when it comes to throughput, the experiments allowed it to perform \emph{thread-unsafe} queries without synchronization (as its original design was not targeting concurrent updates~\cite{TopKapi}). This adaptation clearly favors the throughput performance of Topkapi, which would otherwise require a synchronization mechanism. However, it allows us to establish a best-case estimate for Topkapi's performance under parallelism and compare it with our approach. This does, however, not affect the accuracy evaluation of Topkapi since a single query is carried out at the end of an element stream, irrespective of the synchronization mechanisms used. 

\medskip

\noindent\textbf{Experiment Parameters:} 
Across experiments, the following parameters are varied: \emph{skewness} of the input distribution, \emph{frequent element threshold parameter ($\phi$, controlling query size)}, \emph{number of dispatched threads}, \emph{query rate}, and \emph{stream length}. The latter simulates the point in the execution when a query occurs.
To keep the vast parameter space minimal, $\epsilon$, present in all baselines, is set to $\epsilon=0.1\phi$. The PRIF-specific $\beta$ parameter controlling the delay at which elements are sent 
to the merging thread is set to $\beta=0.9\epsilon$, as in the authors' evaluation \cite{zhang_efficient_2014}. The Topkapi-specific \emph{rows} (also present in the Count-Min Sketch \cite{cormode_improved_2005}) parameter controlling the probability of failure to estimate an element count within a certain error is affixed to 4, which was also the case in the authors' evaluation~\cite{TopKapi}. 
The analysis in Section~\ref{sec:algo_impl_analysis} implies that \shortalgorithmname finds all frequent elements with $\frac{1}{\epsilon}$ counters for both the real-world data sets and the Zipf distribution data sets with $a\leq1$. However, for Zipf data sets with $a>1$, $\frac{1}{\epsilon T}^{\frac{1}{a}}$ counters suffice.

\subsection{\subalgorithmname}
We begin our evaluation by examining the impact of the inner algorithm employed by \shortalgorithmname. We study the differences in latency and throughput between \shortsubalgorithmname and Space-Saving since both algorithms have identical accuracy and memory consumption.

Figure \ref{fig:maxheap_optimization} contains the experimental results, where a query is carried out every 10000 updates, and $\phi$ is set to $10^{-4}=10\epsilon$ to limit the parameter space. 

Figures~\ref{fig:opt_skew_throughput} and \ref{fig:opt_skew_latency} show the latency and throughput over varied skew levels. QOSS yields a higher throughput than the baseline for all skew values, and the latency is significantly lower with QOSS (up to 5x lower) for low skew values of 0.5-1.25. The improvement is less noticeable at the higher skew levels in both figures \ref{fig:opt_skew_throughput}  and \ref{fig:opt_skew_latency}. The performance improvement of \shortsubalgorithmname over Space-Saving is significant when the skew level $a < 1$. This is due to \shortsubalgorithmname's more efficient query procedure and the fact that both algorithms use \(\frac{1}{\epsilon}\) counters instead of \(\big(\frac{1}{\epsilon}\big)^{\frac{1}{a}}\) for higher skew levels. Additionally, lower skew levels contain more frequent elements, as shown in Table \ref{table:frequent_elems}.

    \begin{figure}[t!]
    \begin{subfigure}{0.49\columnwidth}
    \centering
    \includegraphics[width=0.8\columnwidth]{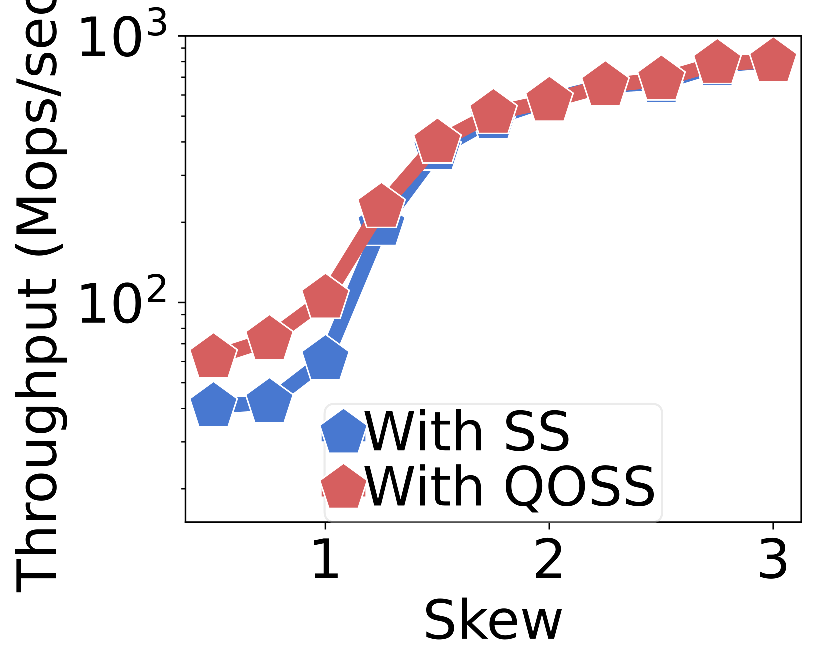}
    \caption{\newtext{Throughput as skew varies  along the x-axis, T=24 threads. Note the logarithmic y-axis 
    }}
    \label{fig:opt_skew_throughput}
    \end{subfigure}
    \begin{subfigure}{0.49\columnwidth}
    \centering
    \includegraphics[width=0.8\columnwidth]{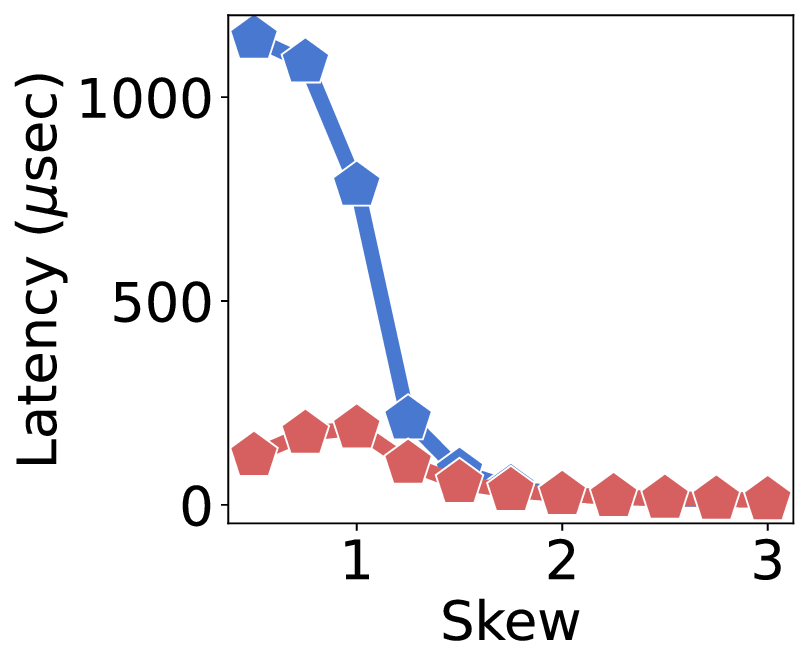}
    \caption{\newtext{Latency as skew varies along the x-axis, T=24 threads}}
    \label{fig:opt_skew_latency}
    \end{subfigure}
    \begin{subfigure}{0.49\columnwidth}
    \centering
    \includegraphics[width=0.8\columnwidth]{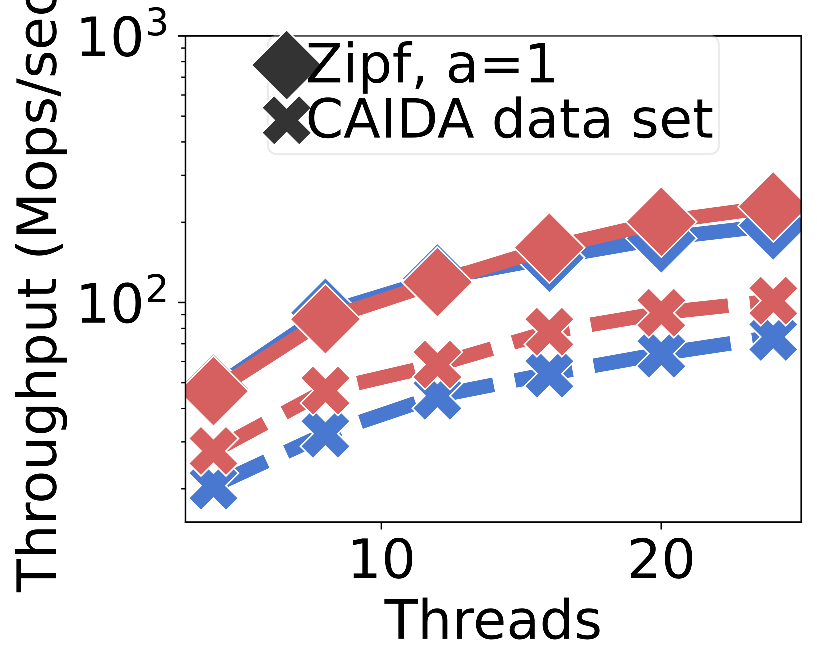}
    \caption{\newtext{Throughput as threads vary  along the x-axis while processing Zipf and CAIDA data sets. Note the logarithmic y-axis.}}
    \label{fig:opt_threads_throughput}
    \end{subfigure}
    \begin{subfigure}{0.49\columnwidth}
    \centering
    \includegraphics[width=0.8\columnwidth]{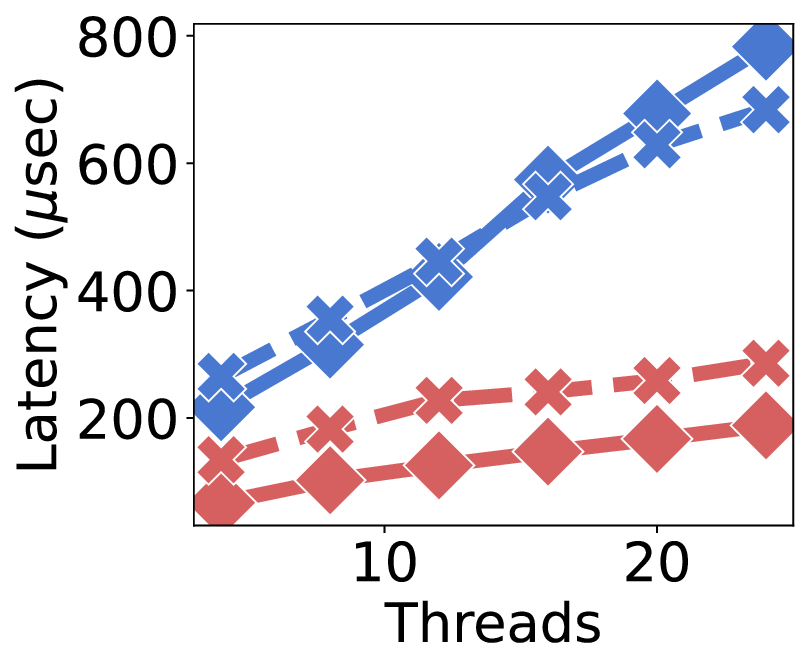}
    \caption{\newtext{Latency as threads vary  along the x-axis while processing Zipf and CAIDA data sets.}}
    \label{fig:opt_threads_latency}
    \end{subfigure}
    \caption{\newtext{Throughput and query latency when \shortalgorithmname employs QOSS or Space-Saving as the inner algorithm. Queries make up 0.01\% of the operations and $\phi=10^{-4}$.}}
    \label{fig:maxheap_optimization}
    \vspace{-3mm}
    \end{figure}

In Figures~\ref{fig:opt_threads_throughput} and~\ref{fig:opt_threads_latency}, the scalability of the approaches is evaluated as the number of dispatched threads varies. The experiments use two data sets: the synthetic Zipf data with skew parameter a=1 and the CAIDA IP-packet trace. The results show that using \shortsubalgorithmname yields up to 1.3x higher throughput for the CAIDA data set; moreover, the improved algorithmic implementation yields lower latency (e.g., up to 4x lower latency for both data sets). In particular, Figure~\ref{fig:opt_threads_latency} shows that the latency of \shortsubalgorithmname scales significantly better with the increasing number of threads, compared to Space-Saving.

The difference between the two methods steadily increases with the number of threads in favor of \shortsubalgorithmname. 
This improvement is attributed to the fact that \shortsubalgorithmname, the query method of which is described in \ref{sec:QOSS}, only needs to check a fraction of its counters to determine which elements have a count above the user-defined threshold $\phi$, while Space-Saving needs to inspect all of its counters.

\subsection{Throughput and Scalability}
To understand how \shortalgorithmname compares to representative approaches Topkapi \cite{TopKapi} and PRIF~\cite{zhang_efficient_2014}, 
we begin by focusing on throughput and scalability when varying the support parameter $\phi$, the number of dispatched threads and the number of concurrent queries. We also compare the internal throughput scalability of \shortalgorithmname (i.e., speedup) to a single-threaded \shortsubalgorithmname algorithm.

Figure~\ref{fig:throughput_skew} shows the update and query throughput relative to the skew of the input data. Three different values of $\phi$ signify queries containing various amounts of frequent elements (see Table \ref{table:frequent_elems}). In each execution, 24 threads are dispatched for each algorithm. The three plots in Figures \ref{fig:sub_throughput_skew_0_queries}, \ref{fig:sub_throughput_skew_0.01_queries}, and \ref{fig:sub_throughput_skew_0.02_queries} contain the results as the number of queries carried out varies.
\shortalgorithmname shows a higher throughput than Topkapi in all cases except for when no queries are carried out, and the skew value is between 0.5 and 1 (figure~\ref{fig:sub_throughput_skew_0_queries}). As the query rate increases from fig \ref{fig:sub_throughput_skew_0_queries} to fig \ref{fig:sub_throughput_skew_0.02_queries}, \shortalgorithmname maintains a high throughput across all skew levels and values of $\phi$, several times higher than Topkapi, which was shown to be highly scaleable in \cite{TopKapi}. As expected, for query rates above 0, the computationally heavy merge operations carried out by Topkapi yield a diminished overall throughput. On the contrary, our algorithms continue to maintain a high throughput despite the presence of concurrent queries.

Due to its query-prioritized design, PRIF copes well with an increased query rate. However, the update throughput of PRIF seems to be strongly dependent on the threshold parameter $\phi$, as setting $\phi=10^{-5}$ yields very low throughput across all skew levels and for all query rates.
Overall, it is observed that \shortalgorithmname is the balanced choice, performing well in most circumstances and combinations of parameter variations.

The plots containing black lines in Figure \ref{fig:throughput_skew} show the speedup relative to a single-threaded \shortsubalgorithmname. The speedup of \shortalgorithmname with 24 dispatched threads compared to a single-threaded \shortsubalgorithmname is around 10-30x across all combinations of $\phi$, query rate, and skew. Interestingly, due to the efficiency of the delegation filters, \shortalgorithmname achieves a speedup greater than the number of dispatched threads in higher skew levels. As worker-threads in \shortalgorithmname swiftly insert elements in filters (an operation consisting of linear searching through a small fixed-size array and incrementing a counter), the single-threaded execution of \shortsubalgorithmname must update a more complex tree data structure, potentially performing multiple time-consuming swap operations to ensure maintained heap properties.
\begin{figure}[t!]
    \centering
    \begin{subfigure}{0.45\textwidth}
    \includegraphics[width=\columnwidth]{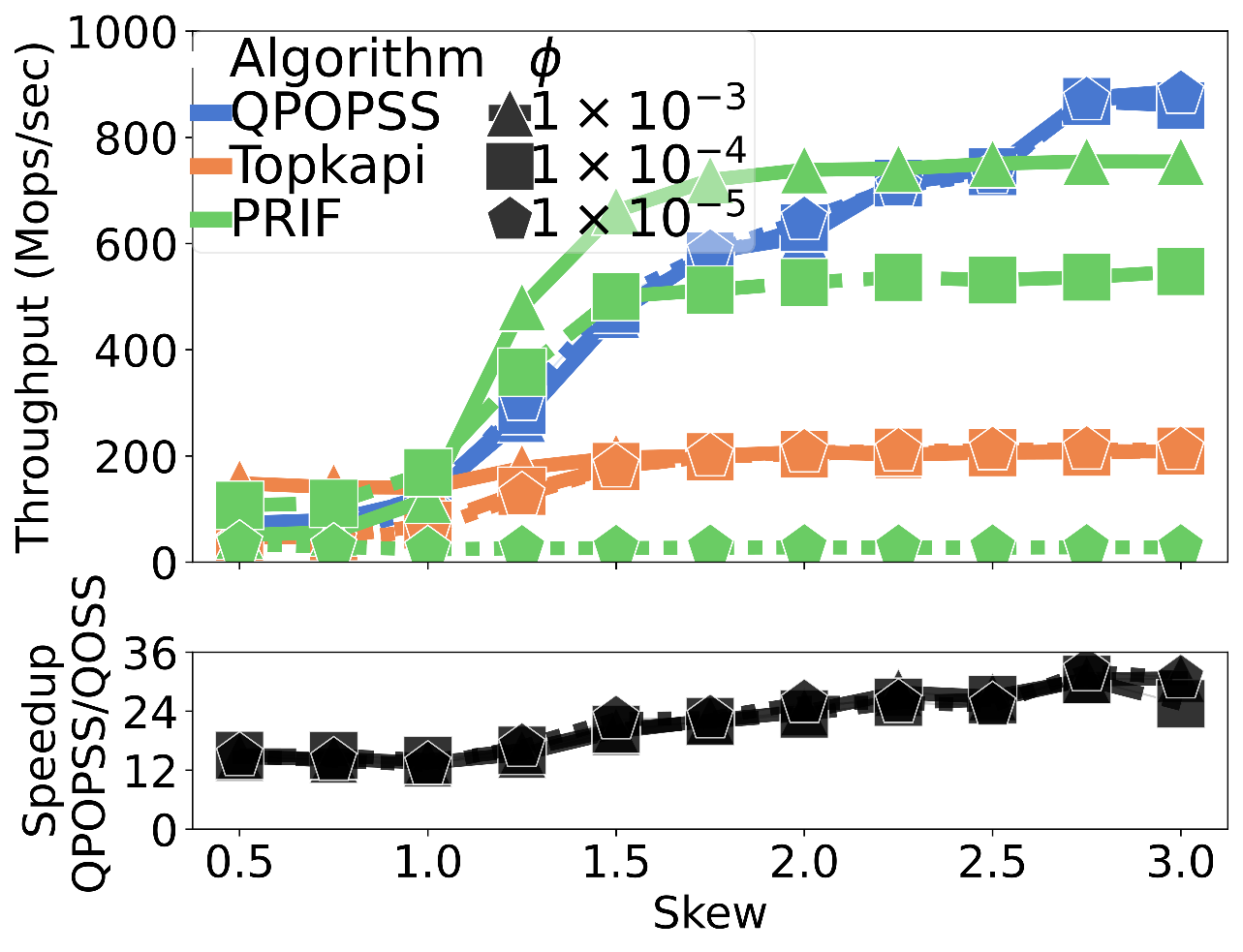}
    \caption{0\% queries.}
    \label{fig:sub_throughput_skew_0_queries}
    \end{subfigure}
    \begin{subfigure}{0.45\textwidth}
    \includegraphics[width=\columnwidth]{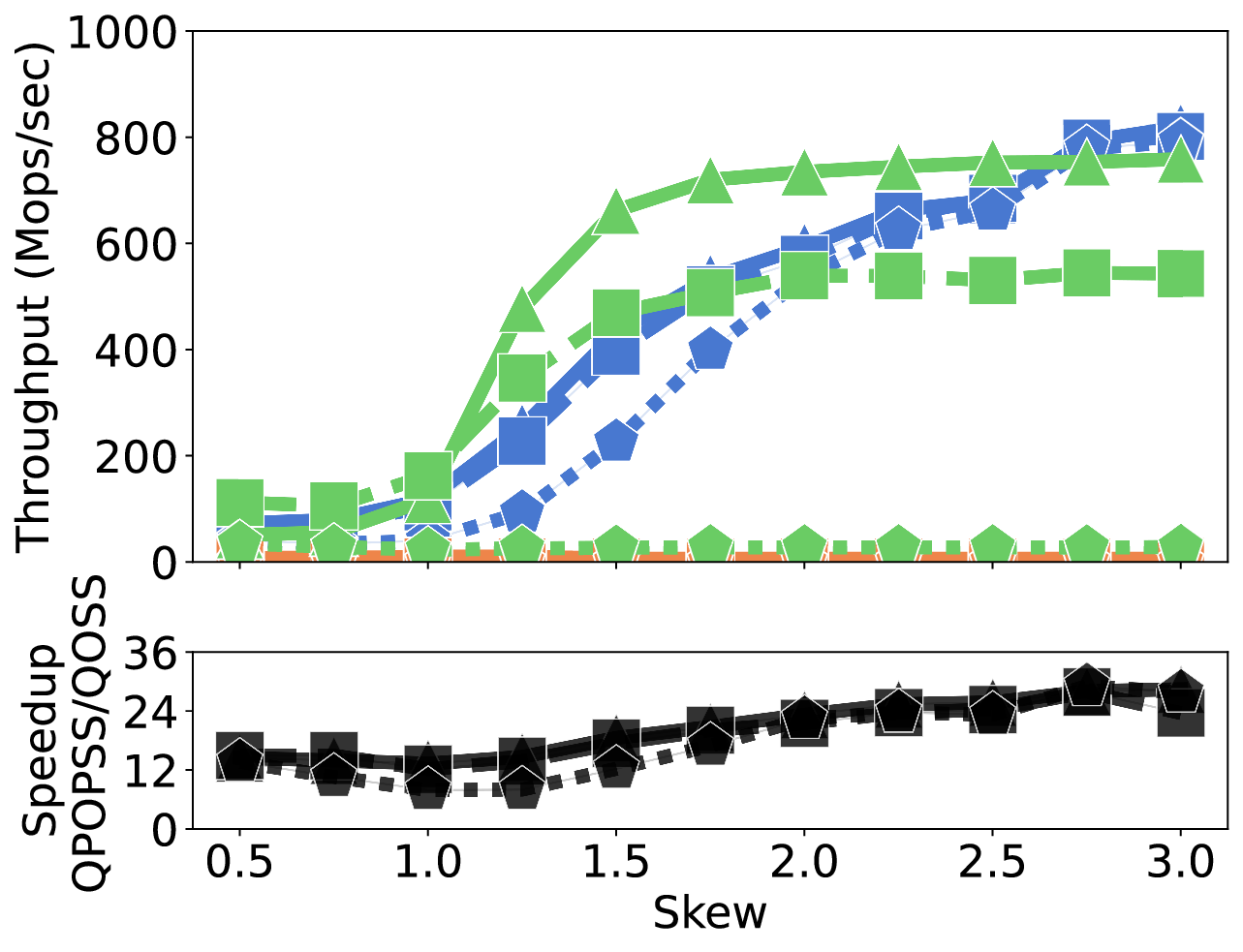}
    \caption{0.01\% queries.}
    \label{fig:sub_throughput_skew_0.01_queries}
    \end{subfigure}
    \begin{subfigure}{0.45\textwidth}
    \includegraphics[width=\columnwidth]{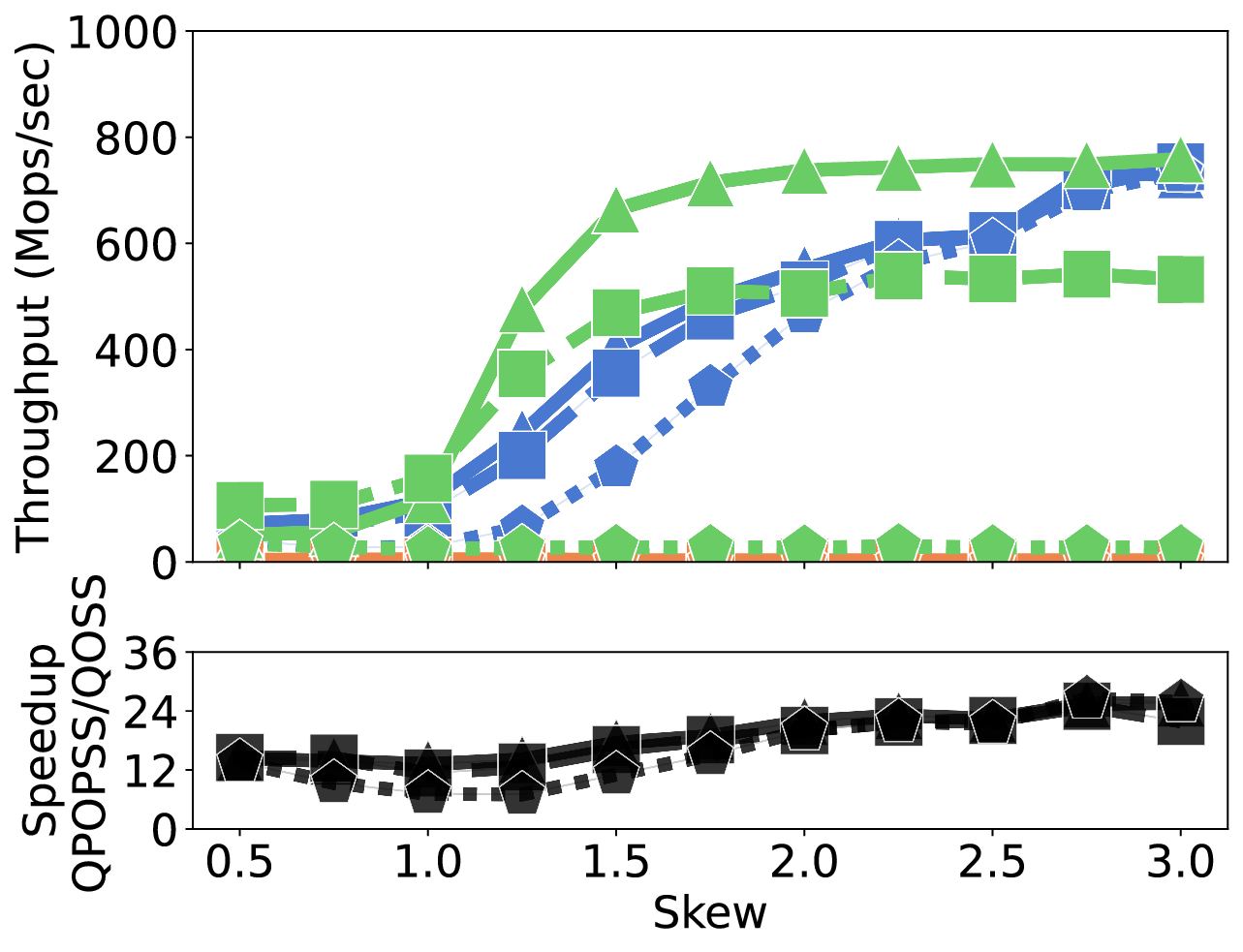}
    \caption{0.02 \% queries.}
    \label{fig:sub_throughput_skew_0.02_queries}
    \end{subfigure}
    \caption{Throughput in million operations per second and multicore speedup of \shortalgorithmname for different skew parameters of the synthetic Zipf data sets. The skew level varies along the x-axis, $T=24$ threads.}
    \label{fig:throughput_skew}
    \vspace{-3mm}
 \end{figure}
 
\begin{figure}[b!]
    \centering
    \begin{subfigure}{0.45\textwidth}
    \includegraphics[width=\columnwidth]{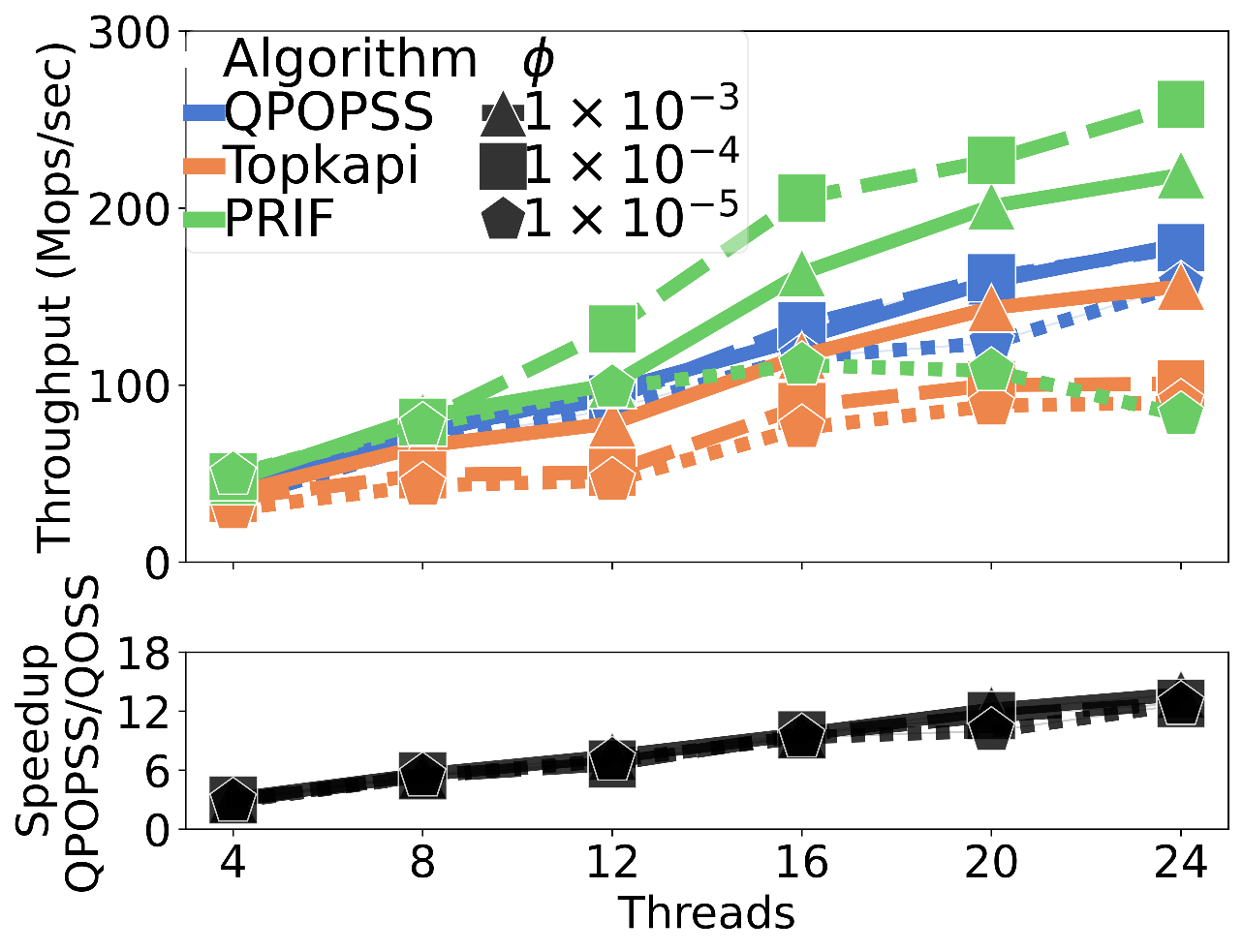}
    \caption{0\% queries.}
    \label{fig:sub_throughput_threads_A_0_queries}
    \end{subfigure}
    \begin{subfigure}{0.45\textwidth}
    \includegraphics[width=\columnwidth]{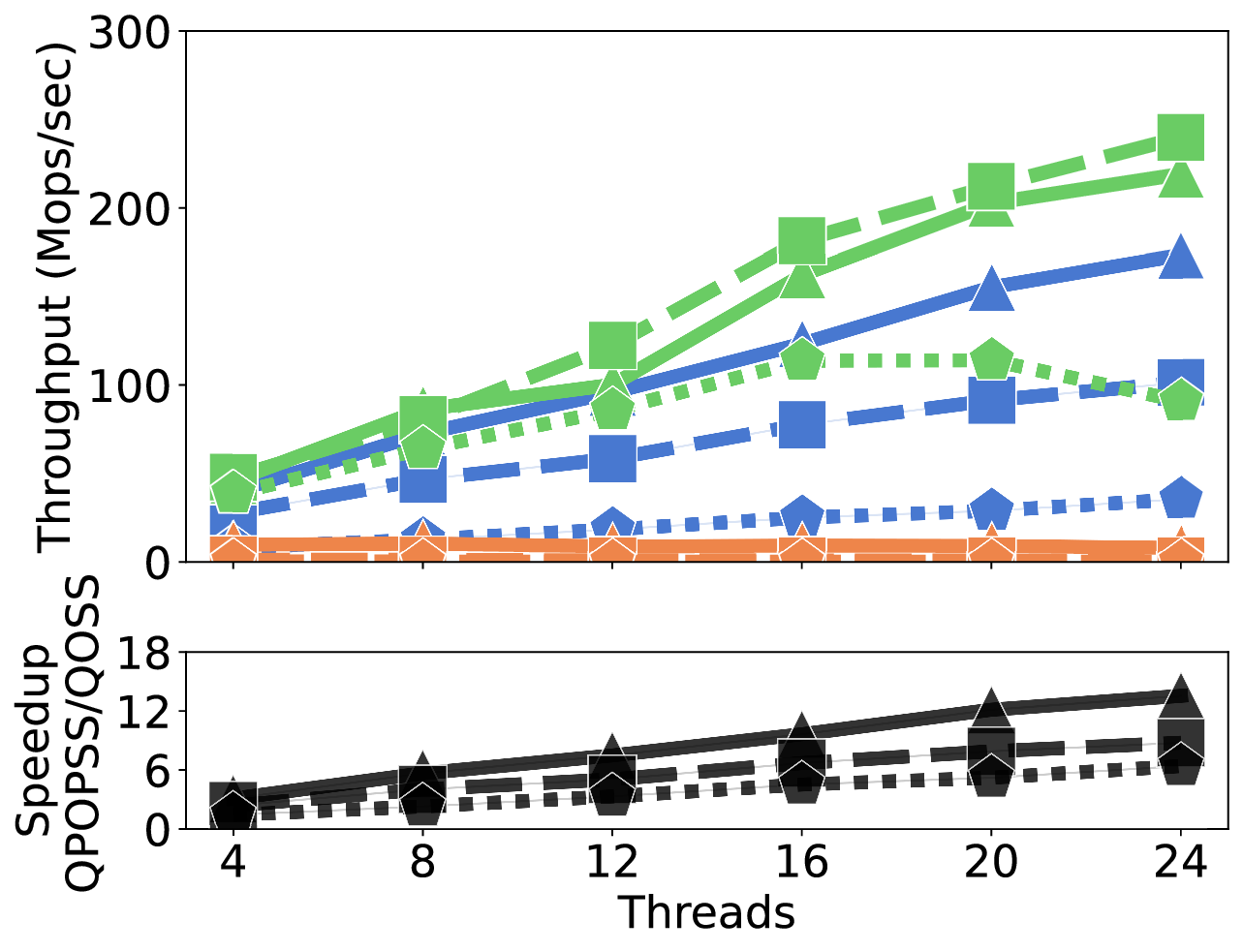}
    \caption{0.01\% queries.}
    \label{fig:sub_throughput_threads_A_0.01_queries}
    \end{subfigure}
    \begin{subfigure}{0.45\textwidth}
    \includegraphics[width=\columnwidth]{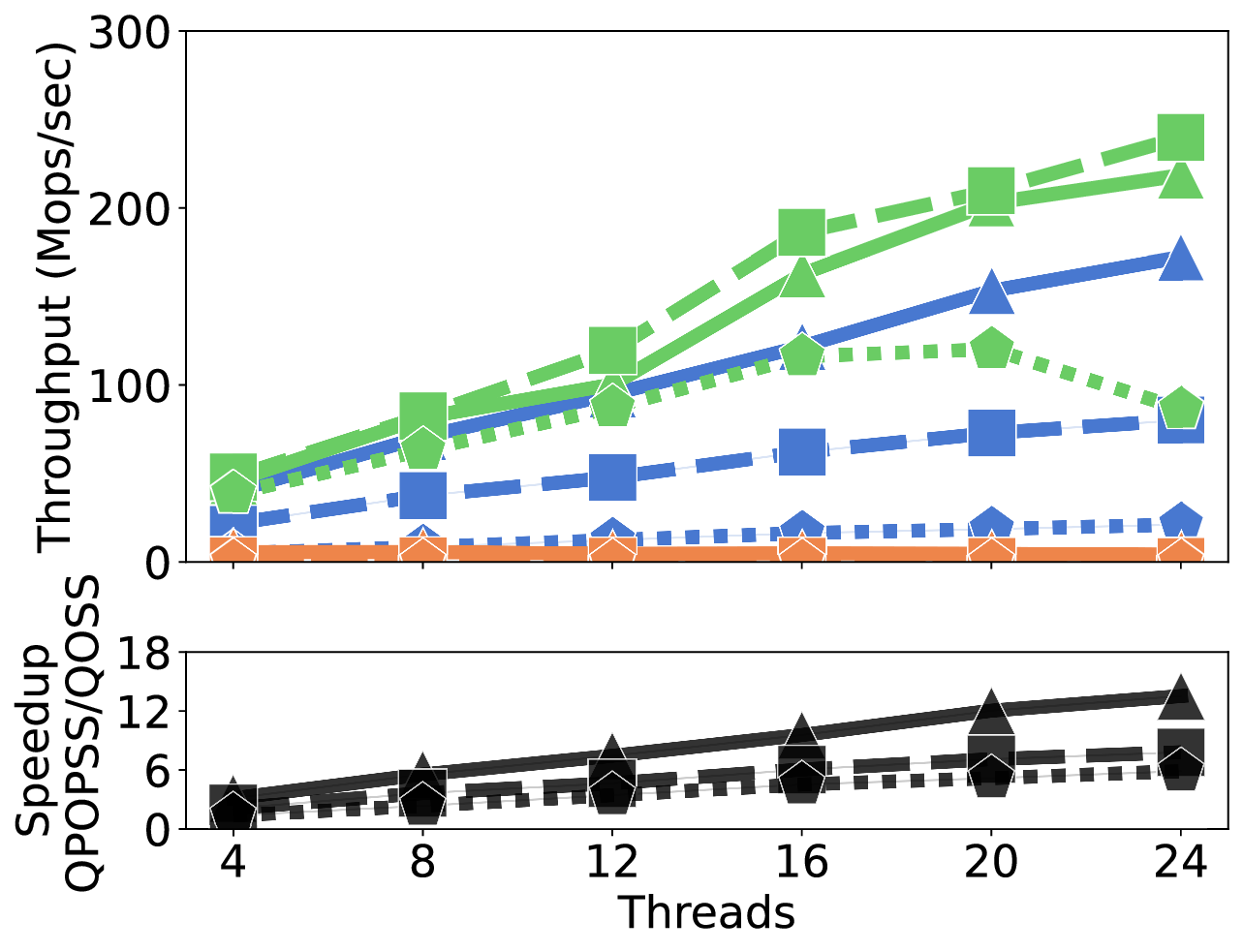}
    \caption{0.02\% queries.}
    \label{fig:sub_throughput_threads_A_0.02_queries}
    \end{subfigure}
    \caption{Throughput in million operations per second and multicore speedup of \shortalgorithmname using the CAIDA backbone router data set. The number of threads varies along the x-axis.}
    \label{fig:throughput_threads_A}
    \vspace{-3mm}
 \end{figure}
The throughput and scalability results for the CAIDA data set are illustrated in Figure~\ref{fig:throughput_threads_A}. Here, the number of threads is varied along the x-axis. 
When no queries are carried out (Figure \ref{fig:sub_throughput_threads_A_0_queries}), the throughput of Topkapi and \shortalgorithmname is similar. However, as the query rate increases, Topkapi's performance rapidly decreases, highlighting the cumbersomeness of the query process. As for PRIF, the throughput does not vary with the query rate but instead depends on the support parameter $\phi$. The different algorithms' throughput clearly correlates to the frequent elements per value of $\phi$ in  Table~\ref{table:frequent_elems}, with numerous frequent elements corresponding to lower throughput and vice-versa. In the case that $\phi=10^{-3}$, PRIF outperforms \shortalgorithmname. Still, in all other cases, \shortalgorithmname can be observed to be the more balanced approach, maintaining high throughput when responding to large and small queries. \shortalgorithmname has a positive trend, wherein throughput increases with the number of dispatched threads. This is not the case for PRIF, which, especially for low values of support parameter $\phi$, seems to have declining throughput as more threads are added. 
This may result from PRIF's single merging thread, which can become a bottleneck.

In the case of real-world data input as in Figure~\ref{fig:throughput_threads_A}, the throughput speedup of \shortalgorithmname scales linearly with the number of threads compared to single-threaded Space-Saving, independently of $\phi$ and the query rate. Interestingly, the throughput of the different algorithms in Figures~\ref{fig:sub_throughput_threads_A_0.01_queries} and~\ref{fig:sub_throughput_threads_A_0.02_queries}
correlates very well to the frequent elements for each data set and value of $\phi$ in Table~\ref{table:frequent_elems}, with numerous frequent elements corresponding to lower throughput and vice-versa.

\subsection{Memory Consumption and Query Accuracy}
\label{sec:evaluation_accuracy_and_space}
We now compare each approach's memory requirements and the ability to report the frequent elements of a stream correctly. Correctness is measured by the metrics recall, precision, and average relative error, previously mentioned in \ref{sec:experimental_setup}. Due to the space-accuracy trade-off associated with the $\epsilon$-approximate $\phi$-frequent problem, we set the memory consumption of each approach to be equal to that of the respective analysis \cite{TopKapi,zhang_efficient_2014} (for \shortalgorithmname, see Corollaries \ref{claim:qpopss_space_requirement} and \ref{claim:qpopss_zipf_requirements}). Each counter equals 32 bytes.

Figure \ref{fig:memory_consumed} 
shows the megabytes consumed by each approach as the number of dispatched threads increases. Three values of $\phi$ are plotted for each baseline. 

\begin{figure}[b!]
\centering
    \includegraphics[width=0.5\columnwidth]{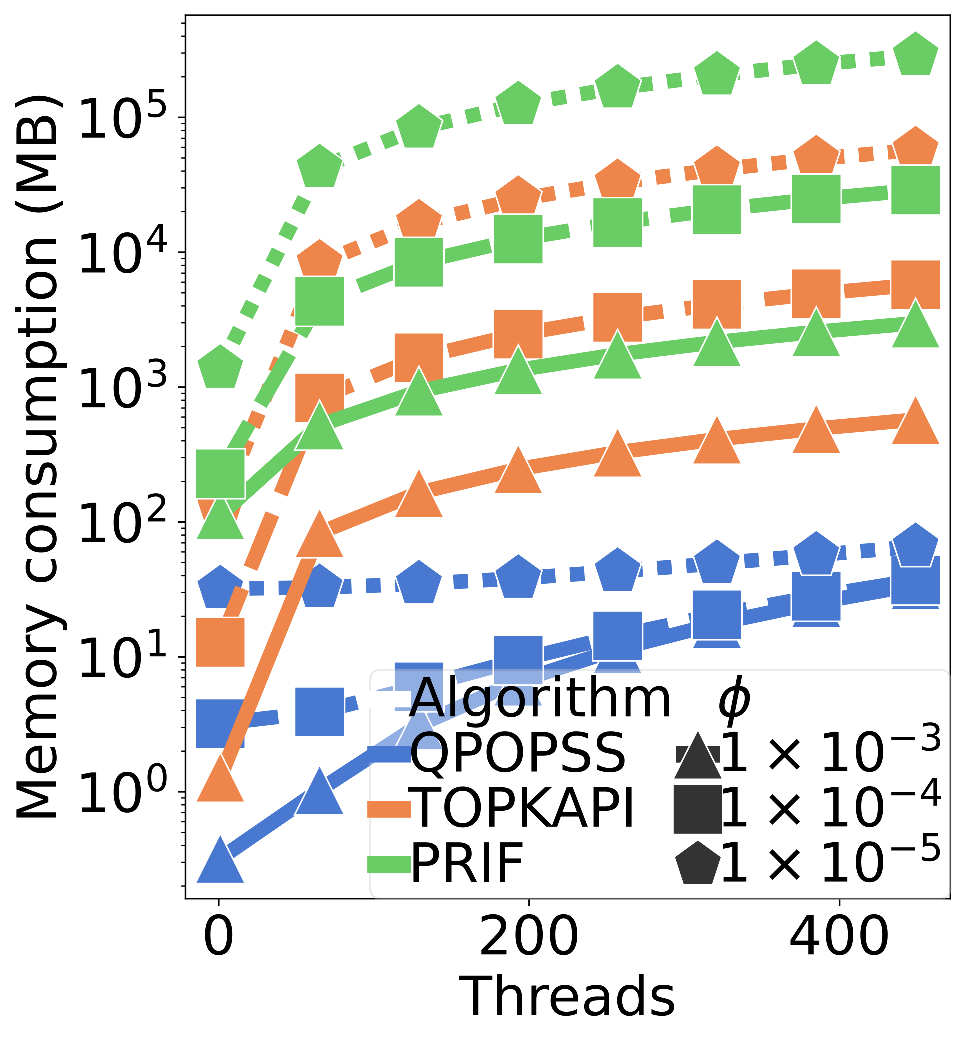}
    \caption{\newtext{Memory consumed by each approach in megabytes. The number of threads varies along the x-axis. Note the logarithmic y-axis.}}
    \label{fig:memory_consumed}
\end{figure}

\shortalgorithmname consumes the least space for each support parameter value of $\phi$ and scales up to 450 threads with at most 65,8 MB consumed in the case of $\phi=10^{-5}$, compared to the 288,1 GB required by PRIF and 57 GB required by Topkapi. The memory consumption of \shortalgorithmname also scales very well with different values of $\phi$, as seen in Figure \ref{fig:memory_consumed}. For example, at 450 threads, $\phi=10^{-3}$ requires 34 MB, $\phi=10^{-4}$ requires 37 MB, and $\phi = 10^{-5}$ requires 65 MB, which is comparatively very low in terms of bytes, while also being highly scalable. 
The PRIF memory requirements are $2\frac{T+1}{\epsilon-\beta}$, where $\beta < \epsilon$, and $T$ is the number of dispatched threads (compared to \shortalgorithmname, the latter uses $\frac{1}{\epsilon}$ counters, with an additional $T^2D$ counters Delegation Filters, where $D$ is the maximum number of unique elements in a filter).

\begin{figure}[t!]
    \begin{subfigure}{0.48\columnwidth}
    \centering
    \includegraphics[width=0.9\columnwidth]{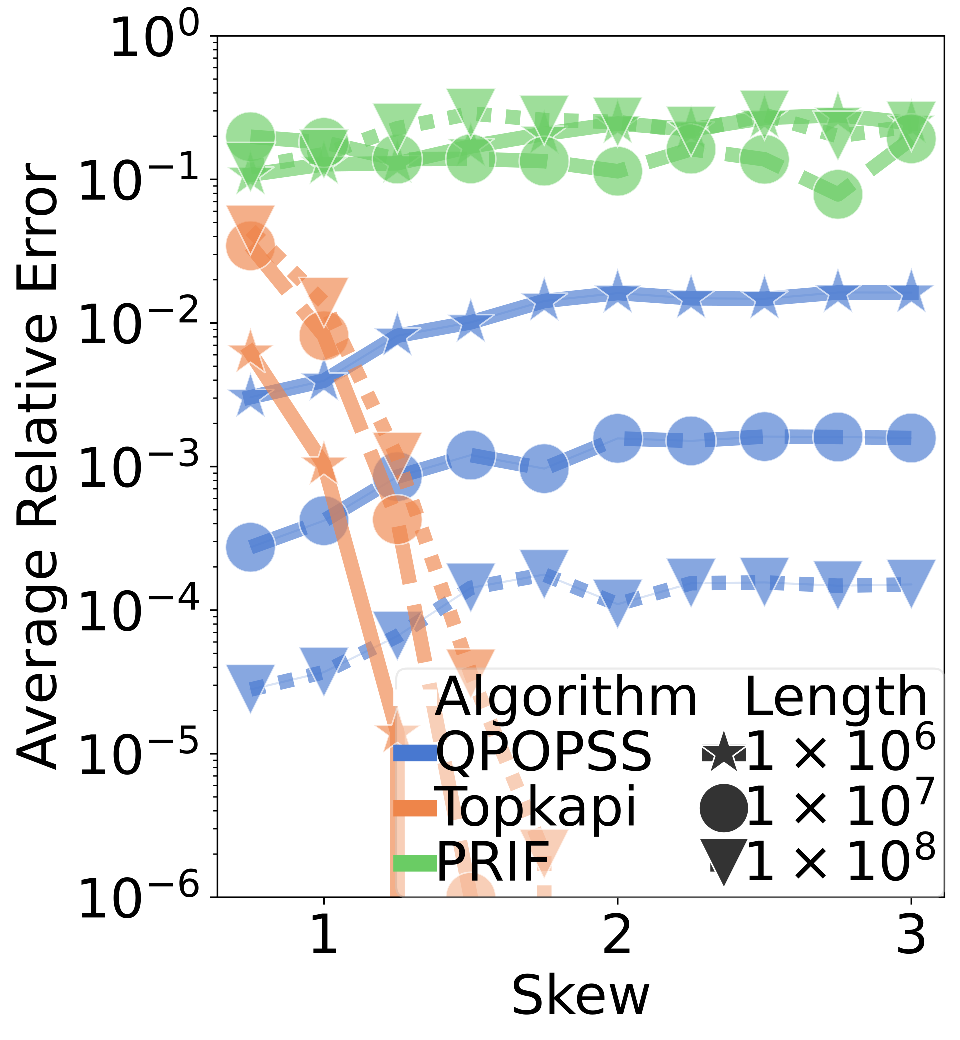}
    \caption{\newtext{Different stream lengths as Zipf skew varies along the x-axis, $T=24$ Threads.}}
    \label{fig:avgre_varyN}
    \end{subfigure}
    \begin{subfigure}{0.48\columnwidth}
    \centering
    \includegraphics[width=0.9\columnwidth]{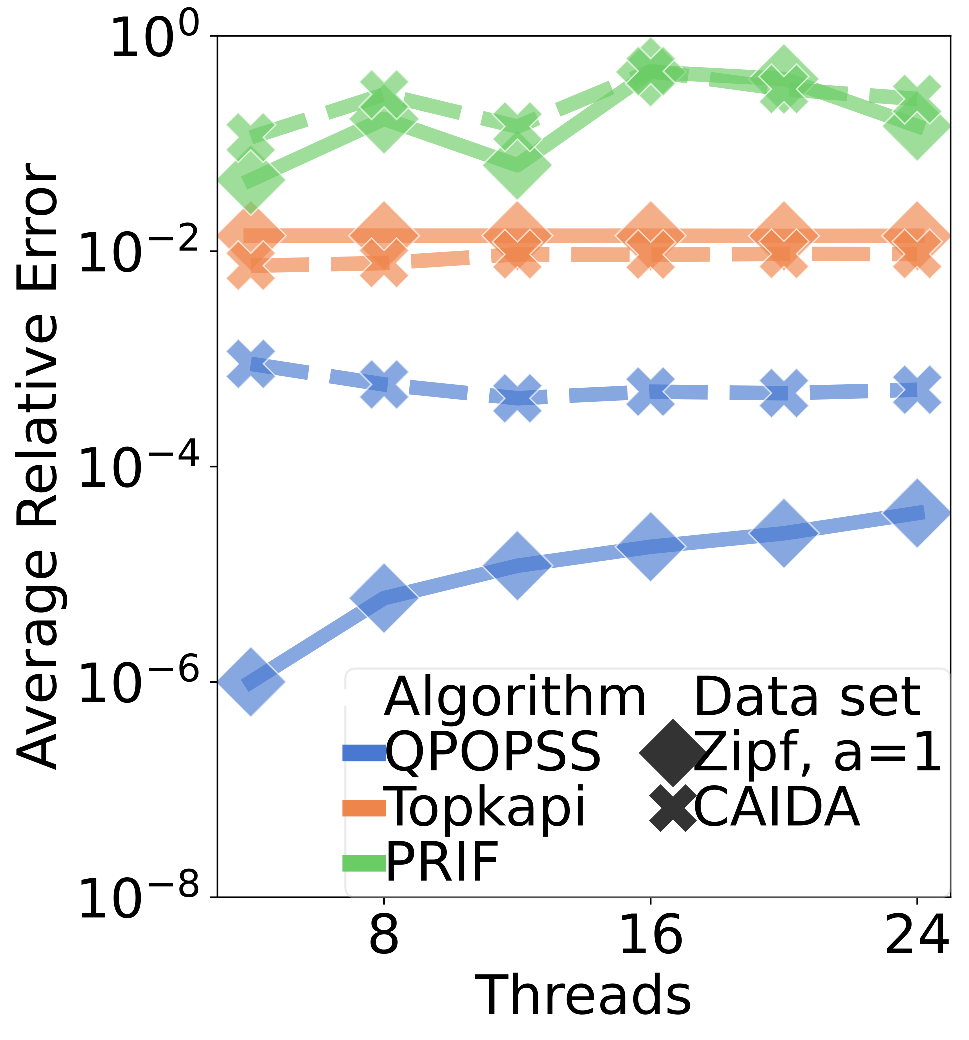}
    \caption{\newtext{Different data sets as the number of threads vary along the x-axis.}}
    \label{fig:avgre_threads}
    \end{subfigure}
    \caption{\newtext{Average relative error. The Zipf skew level varies along the x-axis. $\phi=10^{-4}$. Note the logarithmic y-axes.}}
    \label{fig:average_relative_error}
    \vspace{-3mm}
\end{figure}

\begin{figure}[b!]
    \begin{subfigure}{0.48\columnwidth}
    \centering
    \includegraphics[width=\columnwidth]{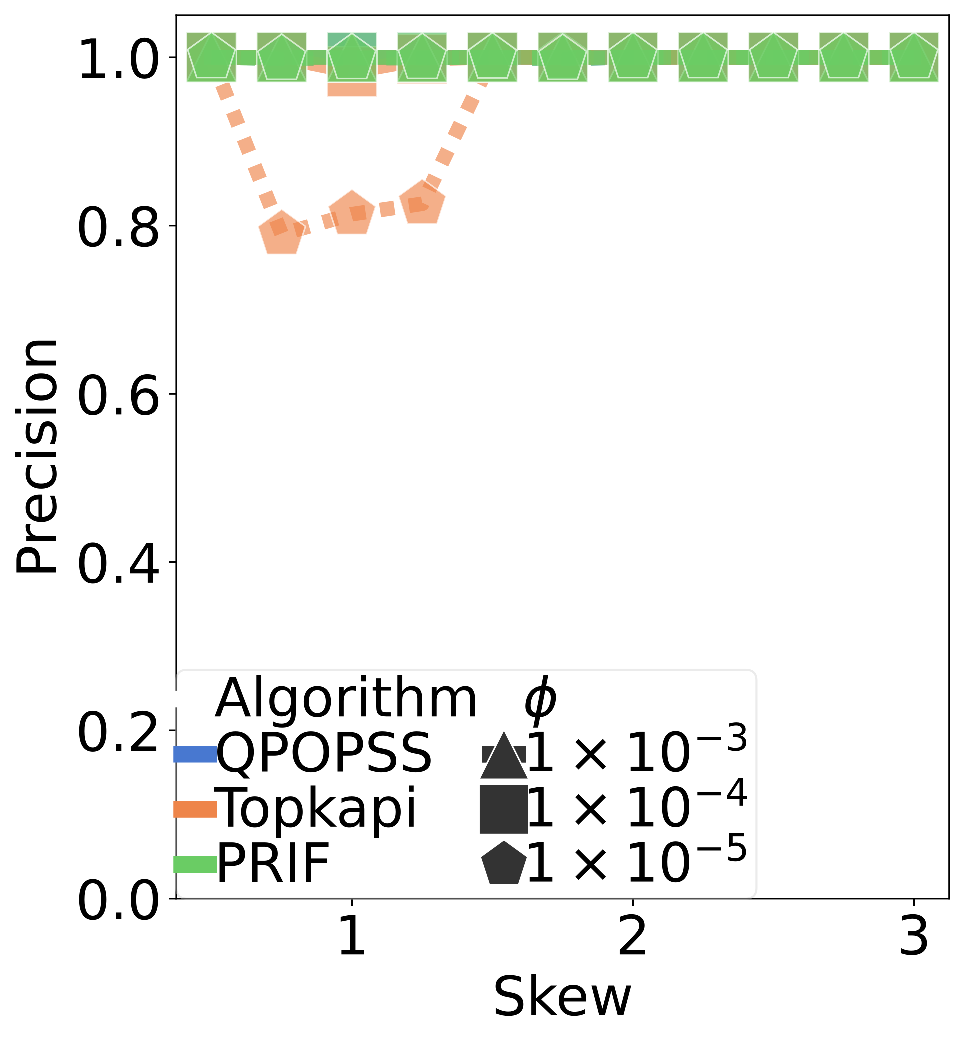}
    \caption{\newtext{Precision.}}
    \label{fig:skew_precision}
    \end{subfigure}
    \begin{subfigure}{0.48\columnwidth}
    \centering
    \includegraphics[width=\columnwidth]{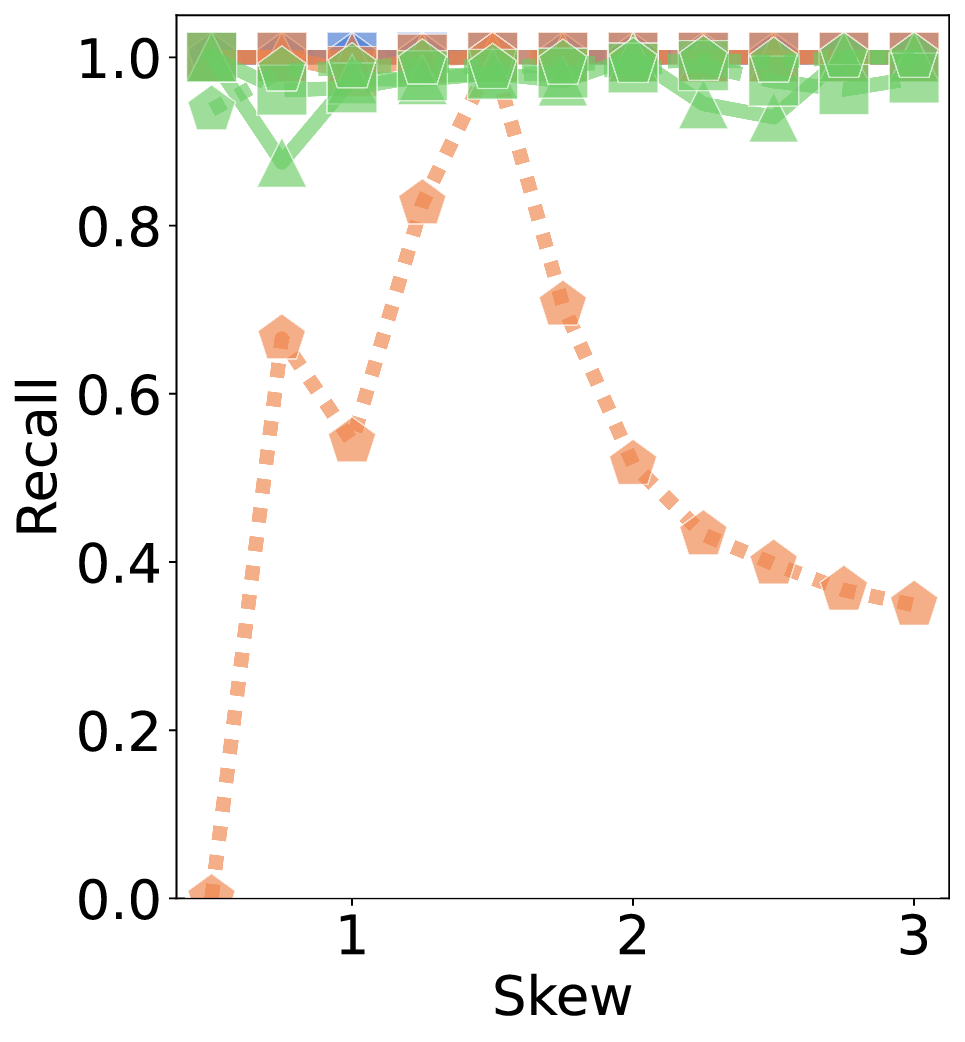}
    \caption{\newtext{Recall.}}
    \label{fig:skew_recall}
    \end{subfigure}
    \caption{\newtext{Accuracy in terms of precision and recall for different threshold values of $\phi$.  The Zipf skew level varies along the x-axis. 
    }}
    \label{fig:precision_and_recall}
    \vspace{-2mm}
\end{figure}
Due to the space/accuracy tradeoff associated with synopsis algorithms, \shortalgorithmname can, therefore, also be said to be very accurate. 
Using the above-described amount of memory bytes, we now compare the accuracy of each approach. To compare fairly between baselines, the experiments entail processing a stream of elements followed by a single query, which is compared to the ground truth. The following results, therefore, show the accuracy of each approach without considering the effects of concurrency. 
Figure~\ref{fig:average_relative_error} contains the average relative error, which is the arithmetic mean of the error in each element count divided by the actual count. The average relative error was measured for different data sets, data stream lengths, and number of dispatched threads. In Figure~\ref{fig:avgre_varyN}, the level of input data skew is varied along the x-axis. The stream lengths are simulated by querying after a certain number of elements have been observed. As the length of the stream increases, \shortalgorithmname displays decreasing average relative error over all skew levels.
This behavior is predicted by theorem~\ref{claim:perfectrecall_zipf}, i.e., as the stream length tends to infinity, the average relative error tends to 0. PRIF's average relative error is relatively high compared to the other baselines and seems to not correlate with either skew level or stream length. 
For Topkapi, however, the average relative error seems to depend less on stream length and more on the skew level as the accuracy improves in the higher levels. 
For many skew values over 2, Topkapi reached 0 average relative error, meaning that all the reported element counts were correct. 
Figure~\ref{fig:avgre_threads} contains the average relative error for two input data sets while varying the number of dispatched threads. The data sets were Zipf with skew level 1 and for the real CAIDA data set. \shortalgorithmname is the approach with the least average relative error, both in the case of real and synthetic data. However, for the synthetic data, there seems to be an upward trend as the number of dispatched threads increases, while for the real data, the average relative error seems to stabilize from 16 to 24 threads. Topkapi and PRIF have a high average relative error, which remains somewhat stable as the number of threads increases. 

The results presented in figure~\ref{fig:precision_and_recall} describe how the baselines compare on precision and recall when the support $\phi$ and skew level of the synthetic data sets are varied. The results show that \shortalgorithmname maintains perfect precision and recall in all cases.
All approaches show high precision and recall; however, \shortalgorithmname is alone in achieving perfect precision and recall for all parameter combinations.
Figure \ref{fig:skew_precision} shows that Topkapi achieves sub-optimal precision in skew levels between 0.75 and 1.25. This can be attributed to Topkapi being based on the probabilistic Count-Min Sketch. In Figure \ref{fig:skew_recall}, both PRIF and Topkapi have varied outcomes, with PRIF's recall dropping to 0.87 in one case.

\subsection{Query Latency}
\label{sec:evaluation_latency}
To understand the delay between issuing a query and the associated response, the execution timing of a series of queries was recorded to calculate a mean value representative of a typical query duration.
In each run of a query latency experiment, the query and update workload is distributed evenly across all threads, with 0.01\% of the operations carried out being queries and the other 99.9\% being update operations. For the experiments in Figure~\ref{fig:query_latency},  the support parameter was set to $\phi = 10^{-4} = 10\epsilon$. A synthetic Zipf data set with shape parameter $a=1$ was used as input to the baselines.

As seen in Figure~\ref{fig:query_latency_skew}, PRIF takes the least time to perform a query in all showcased parameter combinations. This is a result of its query-dedicated merging thread. A design that foregoes memory conservativeness in favor of minimal query latency. A query by \shortalgorithmname takes on the order of 10 to 100s of microseconds, suiting many real-world applications. The Topkapi approach takes on the order of 100s of milliseconds, which introduces delays unsuitable for real-world high-throughput applications. The query latency of \shortalgorithmname decreases with skew level, while the latency of the two other approaches is constant regardless of stream distribution. This is most likely due to the small number of required counters for \shortalgorithmname in the higher skew levels.

\begin{figure}[t!]
    \begin{subfigure}{0.49\columnwidth}
    \includegraphics[width=\columnwidth]{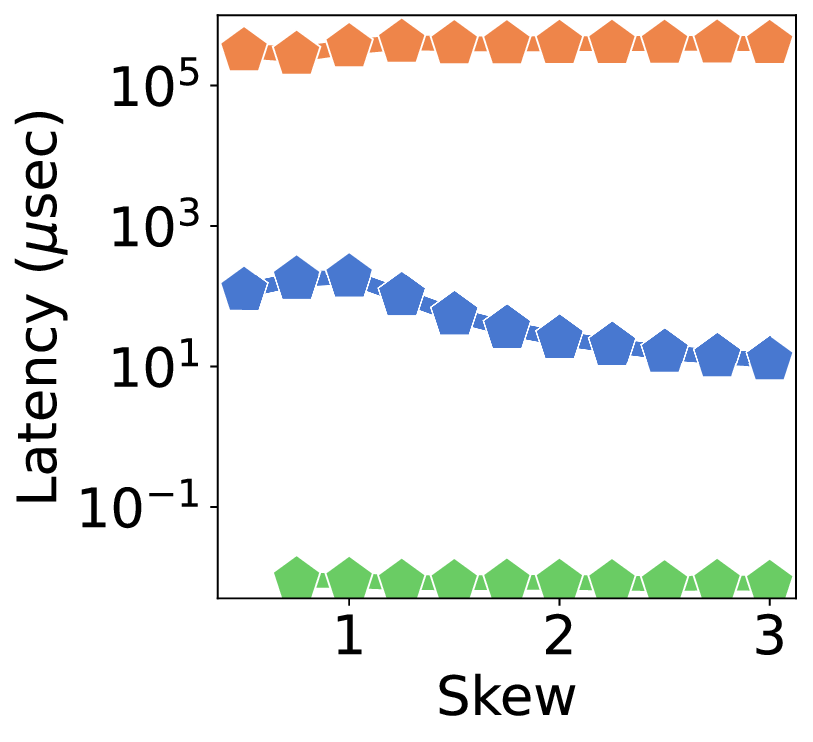}
    \caption{Latency as the Zipf skew level varies along the x-axis. T=24 threads.}
    \label{fig:query_latency_skew}
    \end{subfigure}
    \begin{subfigure}{0.49\columnwidth}
    \includegraphics[width=\columnwidth]{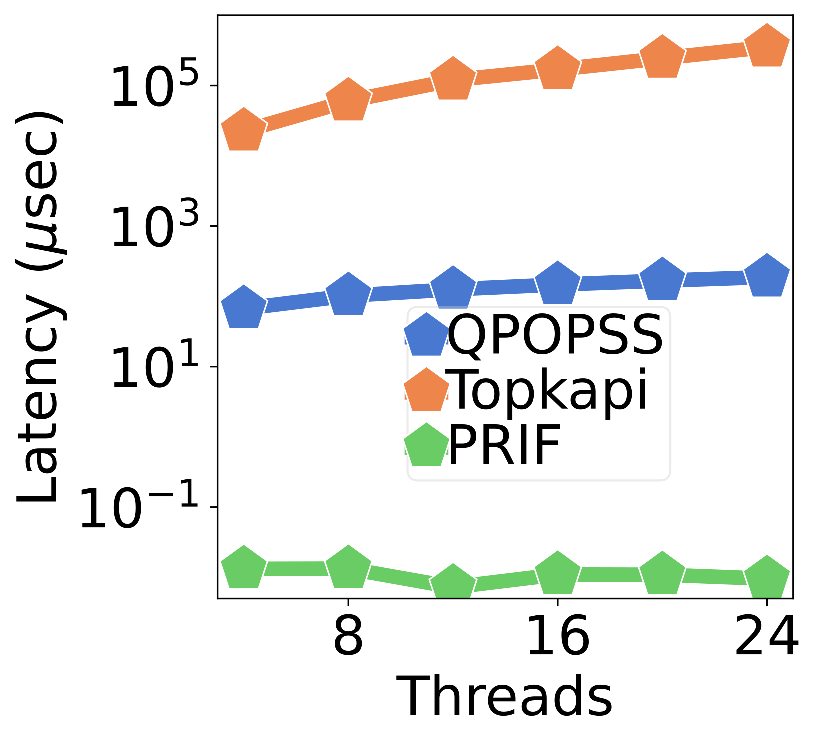}
    \caption{Latency as the number of threads vary along the x-axis}
    \label{fig:query_latency_threads}
    \end{subfigure}
    \caption{Query latency when 0.01\% queries are carried out and $\phi=10^{-4}$. Note the logarithmic y-axes. }
    \label{fig:query_latency}
    \vspace{-3mm}
\end{figure}

Figure~\ref{fig:query_latency_threads} contains the results of the experiments where the number of threads was varied. These results are consistent with previous ones. Due to the design of PRIF, where a single merging thread is queried, it is not affected by increasing the number of threads. \shortalgorithmname has a slight upward trend, meaning that the number of threads dispatched affects the query latency since more \shortsubalgorithmname instances need to be merged. Topkapi has a slightly steeper upward trend due to its cumbersome merging process, initiated each time a query is carried out. Nonetheless, \shortalgorithmname maintains at least an order of magnitude lower latency compared to Topkapi across all numbers of threads. 

\subsection{Summary}
When it comes to the comparison between Space-Saving and \shortsubalgorithmname as the inner algorithm employed by QPOPSS, throughput and latency greatly improve, especially in the lower skew levels.
The results of the comparative evaluation between the state-of-the-art methods are summarized in Table \ref{table:results_summary}. The throughput of \shortalgorithmname excels when processing streams with high data skew and while answering large queries. Compared to PRIF and Topkapi, \shortalgorithmname handles this scenario exceptionally well. When processing the CAIDA data set, especially when queries are present in the workload, Topkapi lacks the competitive throughput of \shortalgorithmname and PRIF\footnote{recall that this is despite the fact that  Topkapi was given an advantage regarding throughput in the presence of concurrent queries, by not enforcing thread-safe synchronization, as they were not part of its design}. PRIF handles both the low and high query rates equally well due to its query-favored design. The precision and recall of \shortalgorithmname are perfect when processing Zipf data sets, and the average relative error (ARE) is low, diminishing quickly as the stream length grows. PRIF and Topkapi have higher ARE when processing real-world data, and Topkapi has excellent ARE when processing high-skew synthetic data. When scalable memory consumption is important, \shortalgorithmname has a clear advantage over both PRIF and Topkapi due to their counters increasing with a factor of $T$, which is not the case for \shortalgorithmname.
PRIF excels in latency due to its design of constant merging by a dedicated thread, while \shortalgorithmname and Topkapi appear less favorable due to employing a merge-on-demand style of querying.

\begin{table}[htb]
\centering
\footnotesize\setlength{\tabcolsep}{6pt}
\begin{tabular}{|l|ll|llllll|}
\hline
\multicolumn{1}{|c|}{\multirow{2}{*}{\begin{turn}{90}Baseline\end{turn}}} & \multicolumn{2}{l|}{\multirow{2}{*}{Query aspects}} & \multicolumn{6}{c|}{Metric}                                                                                                                                                                                                                                    \\[5ex] \cline{4-9} 
\multicolumn{1}{|c|}{}                                                     & \multicolumn{2}{l|}{}                                  & \multicolumn{1}{c|}{Throughput}              & \multicolumn{1}{c|}{Precision}           & \multicolumn{1}{c|}{Recall}                  & \multicolumn{1}{c|}{ARE}                 & \multicolumn{1}{c|}{Memory}                  & \multicolumn{1}{c|}{Latency} \\ \hline
\multirow{4}{*}{\vspace{-0.5em}\begin{turn}{90}QPOPSS\end{turn}}                         & \multicolumn{1}{l|}{\multirow{2}{*}{Few}}    & Small   & \multicolumn{1}{l|}{$\uparrow \uparrow$}     & \multicolumn{1}{l|}{$\uparrow \uparrow$} & \multicolumn{1}{l|}{$\uparrow \uparrow$}     & \multicolumn{1}{l|}{$\uparrow \uparrow$} & \multicolumn{1}{l|}{$\uparrow \uparrow$}     & $\uparrow$                   \\ \cline{3-9} 
                                                                           & \multicolumn{1}{l|}{}                        & Large   & \multicolumn{1}{l|}{$\uparrow$}              & \multicolumn{1}{l|}{$\uparrow \uparrow$} & \multicolumn{1}{l|}{$\uparrow \uparrow$}     & \multicolumn{1}{l|}{$\uparrow \uparrow$} & \multicolumn{1}{l|}{$\uparrow \uparrow$}     & $\uparrow$                   \\ \cline{2-9} 
                                                                           & \multicolumn{1}{l|}{\multirow{2}{*}{Many}}   & Small   & \multicolumn{1}{l|}{$\uparrow$}              & \multicolumn{1}{l|}{$\uparrow \uparrow$} & \multicolumn{1}{l|}{$\uparrow \uparrow$}     & \multicolumn{1}{l|}{$\uparrow \uparrow$} & \multicolumn{1}{l|}{$\uparrow \uparrow$}     & $\uparrow$                   \\ \cline{3-9} 
                                                                           & \multicolumn{1}{l|}{}                        & Large   & \multicolumn{1}{l|}{$\uparrow \uparrow$}     & \multicolumn{1}{l|}{$\uparrow \uparrow$} & \multicolumn{1}{l|}{$\uparrow \uparrow$}     & \multicolumn{1}{l|}{$\uparrow \uparrow$} & \multicolumn{1}{l|}{$\uparrow \uparrow$}     & $\uparrow$                   \\ \hline
\multirow{4}{*}{\begin{turn}{90}PRIF\end{turn}}                           & \multicolumn{1}{l|}{\multirow{2}{*}{Few}}    & Small   & \multicolumn{1}{l|}{$\uparrow \uparrow$}     & \multicolumn{1}{l|}{$\uparrow \uparrow$} & \multicolumn{1}{l|}{$\uparrow$}              & \multicolumn{1}{l|}{$\downarrow$}        & \multicolumn{1}{l|}{$\downarrow \downarrow$} & $\uparrow \uparrow$          \\ \cline{3-9} 
                                                                           & \multicolumn{1}{l|}{}                        & Large   & \multicolumn{1}{l|}{$\downarrow \downarrow$} & \multicolumn{1}{l|}{$\uparrow \uparrow$} & \multicolumn{1}{l|}{$\uparrow \uparrow$}     & \multicolumn{1}{l|}{$\downarrow$}        & \multicolumn{1}{l|}{$\downarrow \downarrow$} & $\uparrow \uparrow$          \\ \cline{2-9} 
                                                                           & \multicolumn{1}{l|}{\multirow{2}{*}{Many}}   & Small   & \multicolumn{1}{l|}{$\uparrow \uparrow$}     & \multicolumn{1}{l|}{$\uparrow \uparrow$} & \multicolumn{1}{l|}{$\uparrow$}              & \multicolumn{1}{l|}{$\downarrow$}        & \multicolumn{1}{l|}{$\downarrow \downarrow$} & $\uparrow \uparrow$          \\ \cline{3-9} 
                                                                           & \multicolumn{1}{l|}{}                        & Large   & \multicolumn{1}{l|}{$\downarrow \downarrow$} & \multicolumn{1}{l|}{$\uparrow \uparrow$} & \multicolumn{1}{l|}{$\uparrow \uparrow$}     & \multicolumn{1}{l|}{$\downarrow$}        & \multicolumn{1}{l|}{$\downarrow \downarrow$} & $\uparrow \uparrow$          \\ \hline
\multirow{4}{*}{\begin{turn}{90}Topkapi\end{turn}}                        & \multicolumn{1}{l|}{\multirow{2}{*}{Few}}    & Small   & \multicolumn{1}{l|}{$\uparrow$}              & \multicolumn{1}{l|}{$\uparrow \uparrow$} & \multicolumn{1}{l|}{$\downarrow \downarrow$} & \multicolumn{1}{l|}{$\uparrow$}          & \multicolumn{1}{l|}{$\downarrow$}            & $\downarrow$                 \\ \cline{3-9} 
                                                                           & \multicolumn{1}{l|}{}                        & Large   & \multicolumn{1}{l|}{$\uparrow$}              & \multicolumn{1}{l|}{$\uparrow$}          & \multicolumn{1}{l|}{$\downarrow$}            & \multicolumn{1}{l|}{$\uparrow$}          & \multicolumn{1}{l|}{$\downarrow$}            & $\downarrow$                 \\ \cline{2-9} 
                                                                           & \multicolumn{1}{l|}{\multirow{2}{*}{Many}}   & Small   & \multicolumn{1}{l|}{$\downarrow \downarrow$} & \multicolumn{1}{l|}{$\uparrow \uparrow$} & \multicolumn{1}{l|}{$\downarrow \downarrow$} & \multicolumn{1}{l|}{$\uparrow$}          & \multicolumn{1}{l|}{$\downarrow$}            & $\downarrow$                 \\ \cline{3-9} 
                                                                           & \multicolumn{1}{l|}{}                        & Large   & \multicolumn{1}{l|}{$\downarrow \downarrow$} & \multicolumn{1}{l|}{$\uparrow$}          & \multicolumn{1}{l|}{$\downarrow$}            & \multicolumn{1}{l|}{$\uparrow$}          & \multicolumn{1}{l|}{$\downarrow$}            &   $\downarrow$                           \\ \hline
\end{tabular}
\caption{Summary of the evaluation results, given a high number of dispatched threads, moderate to high data skewness, and considerable stream length. The number of $\uparrow$ and $\downarrow$ symbols indicate positive and negative comparative performance on a specific metric, respectively (i.e., $\uparrow$ on latency and ARE means a low and therefore desirable metric value).}
\label{table:results_summary}
\end{table}
 \section{Other related work}\label{sec:relatedwork}

As briefly touched upon in Section~\ref{sec:problem_analysis}, several algorithms target variants of the $\epsilon$-approximate frequent elements using a multi-threading, with emphasis on thread-local approaches as discussed in that section and the evaluation baselines~\cite{TopKapi, Augmented_Sketch,zhang_efficient_2014,cafaro_parallel_2016, das_cots:_2009}.

The Augmented Sketch (ASketch) \cite{Augmented_Sketch} is a highly accurate stream processing algorithm for element frequency estimation. The design comprises two interconnected data structures: a sketch and a filter. The fixed-size filter tracks elements and their occurrence. When the filter becomes full and a non-tracked element appears in the stream, the element is inserted both in the filter and in the underlying sketch. The ASketch improves accuracy and increases throughput when processing streams with high data skew. The authors also provide designs for parallel processing with either the filter and the sketch running on different cores or where a complete ASketch runs on a reserved core. However, although frequent elements estimation is possible, the approach focuses on the point estimation of the frequency count of specific elements. 

The HeavyKeeper algorithm~\cite{heavy-keeper} targets combining counter-based and sketch-based approaches for element-count tracking. By periodically decaying element counts, freshness is ensured as input data distributions shift over time. HeavyKeeper can be combined with a min-heap to track the most frequent elements of a stream. It is a sequential algorithm, though, and its parallelization, allowing concurrent updates with queries and the appropriate supporting data structures, is not discussed in the work.

M. Cafaro et al.~\cite{cafaro_parallel_2016} present a parallel design utilizing Space-Saving as the core algorithm in a purely thread-local design. Queries merge each Space-Saving algorithm instance in a tree-like fashion until only a final algorithm instance that contains the frequent elements is left. The design gives perfect speedup compared to a sequential Space-Saving algorithm instance. Moreover, the accuracy of the presented design was precisely equal to the single-threaded version. However, no scheme for concurrent queries is given and is therefore not usable in real-life applications where continuous queries are required, which are targeted here. This fact, coupled with our inability to find a readily available open-source implementation, made us choose not to include the approach in our evaluation.

The {Cooperative Thread Scheduling Framework (CoTS)}~\cite{das_cots:_2009} and its multi-stream extension~\cite{parallel_space_saving} are approaches for parallelizing the frequent elements problem. Similar to \shortalgorithmname, this approach builds on the Space-Saving algorithm.
The design features a single Space-Saving algorithm instance on which each thread operates. Update operations are carried out directly on the space-saving algorithm instance or handed over to whichever thread currently has exclusive access. The synchronization primitives used are lock-free, promoting high throughput. 

The approach was evaluated using synthetic Zipfian data sets, showing relative throughput gains between CoTS and a lock-based design. 
However, the work lacks discussion on the effect of overlapping updates and queries on query accuracy and rely solely on the analysis of the sequential Space-Saving algorithm. Additionally, due to the combination of the complexity of the approach and the absence of a readily available open-source implementation, implying risks of misinterpretation of the work, it was not possible to include this approach in our evaluation.
\section{Conclusions}\label{sec:conclusions}
Data analytics for demanding high-throughput applications (e.g., optimizing network traffic, detecting cyber-security threats, and performing online data analysis) is important in many fields. It requires novel algorithmic solutions that exploit concurrency to achieve a higher degree of parallelism. To this end, we analyzed the problem of finding the frequent elements of high throughput data streams with concurrent updates and queries, identifying a set of challenges. 

Intending to address these challenges, we designed and extensively evaluated \algorithmname, 
both analytically and empirically, exploring the extended trade-off space in the presence of concurrent queries and updates for the frequent elements problem. To address concurrency-associated accuracy challenges we provided a bound on the space required by \shortalgorithmname to report the frequent elements of a stream accurately. Furthermore, we bounded the frequent elements reported and their estimated occurrence when queries overlap concurrently with updates. To our knowledge, this has not been done before in the context of frequent elements estimation. 

In addressing timeliness challenges, we proposed the \subalgorithmname algorithmic implementation, which was shown to reduce the query latency compared to the original approach drastically. We evaluated \shortalgorithmname through comparison with representative methods in the literature, using synthetic and real-life data sets. The results clearly show that \shortalgorithmname accurately and swiftly reports the frequent elements of a stream compared to other approaches, using extremely few bytes of memory, in line with memory-footprint accuracy-associated challenges. This can be attributed to the unique combination of domain partitioning, inter-thread filters, and query-optimized synopsis data structures, which together ensure high throughput, low latency, low memory, and high accuracy for an overall balanced approach.
The code and the data used (alt. code to generate the synthetic data-sets) are openly available~\cite{Delegation_Space-Saving_2022}.

Future studies may build on the work presented in this paper by extending the method beyond frequent element estimation, such as quantiles, histograms, wavelets, and similar.
Moreover, recent advances where parallelization improvements are of interest include methods for finding frequent elements in streams of both updates and removals \cite{space-saving-plus-minus}. Further work is also needed to study efficient mechanisms for tracking the most recent frequent elements \cite{heavy_hitter_windows_basat,memento_basat}, as the stream data distribution may change over time. 

\section*{Acknowledgements}
Partially supported by the following grants:
the Marie Skłodowska-Curie Doctoral Network prj. RELAX-DN, funded by the European Union under Horizon Europe 2021-2027 Framework Programme Grant Agreement nr. 101072456; the Chalmers AoA frameworks Energy and Production, WPs INDEED, and ``Scalability, Big Data and AI", respectively; the Swedish Research Council grant “EPITOME” (VR 2021-05424);  the Wallenberg AI, Autonomous Systems and Software Program and Wallenberg Initiative Materials for Sustainability prj. WASP-WISE STRATIFIER; and the TANDEM project within the framework of the Swedish Electricity Storage and Balancing Centre (SESBC), co-funded by the Swedish Energy Agency.
 
\bibliographystyle{plainurl}
\bibliography{references}

\pagebreak

\pagebreak

\end{document}